\documentclass[10pt]{article}

\usepackage{amsmath,amssymb,amsthm,color,epsfig,mathrsfs}
\usepackage{amsbsy,multirow}

\newtheorem{theorem}{Theorem}

\newtheorem{lemma}[theorem]{Lemma}

\newtheorem{proposition}[theorem]{Proposition}
\newtheorem{remark}[theorem]{Remark}

\newcounter{spslist}

\newcounter{geqncount}
    {\refstepcounter{equation}%
     \setcounter{geqncount}{\value{equation}}%
     \setcounter{equation}{0}%
  }%
    {\setcounter{equation}{\value{geqncount}}}

\newcommand{\Cc}{{\mathcal C}}
\newcommand{\Cco}{\mathring{\mathcal C}}
\newcommand{\go}{\mathring{g}}

\newcommand{\Zo}{\mathring{Z}}
\newcommand{\CC}{\mathbb{C}}

\newcommand{\RR}{\mathbb{R}}

\newcommand{\Spiral}{\mathcal{S}}

\newcommand{\G}{\Gamma}
\newcommand{\Lo}{\mathring{\Lambda}}
\newcommand{\lo}{\mathring{\lambda}}
\newcommand{\So}{\mathring{\Sigma}}

\newcommand{\Mo}{\mathring{M}}
\newcommand{\Do}{\mathring{D}}
\newcommand{\Eo}{\mathring{E}}
\newcommand{\No}{\mathring{N}}

\newcommand{\Ao}{\mathring{A}}
\newcommand{\Bo}{\mathring{B}}

\newcommand{\dom}{\mathcal{D}}

\newcommand{\half}{{\textstyle\frac{1}{2}}}
\newcommand{\quarter}{{\textstyle\frac{1}{4}}}

\newcommand{\Asym}{A^\text{sym}}
\newcommand{\Aosym}{\mathring{A}^\text{sym}}
\newcommand{\cchi}{\text{\raisebox{1.5pt}{$\chi$}}}

\newcommand{\rev}[1]{{\color{black}#1}}

\begin{document}

\bibliographystyle{plain}

\begin{center}
{\bf \Large  Spectra of Regular Quantum Trees:\\\vspace{5pt}
Rogue Eigenvalues and Dependence on Vertex Condition
}
\end{center}

\vspace{0.2ex}

\begin{center}
{\scshape \large Zhaoxia\hspace{-2pt} W\hspace{-2pt}. Hess\footnote{Zhaoxia W\hspace{-2pt}. Hess, Southeastern Louisiana University, Hammond, LA 70402, zhaoxia.hess@selu.edu} \,and\, Stephen P\hspace{-2pt}. Shipman\footnote{Stephen P\hspace{-2pt}. Shipman, Louisiana State University, Baton Rouge, LA 70803, shipman@lsu.edu}} \\
\vspace{1ex}
{\itshape Departments of Mathematics\\
 Southeastern Louisiana University\textsuperscript{1}  \;and\;  Louisiana State University\textsuperscript{2}}
\end{center}

\vspace{3ex}
\centerline{\parbox{0.9\textwidth}{
{\bf Abstract.}\
We investigate the spectrum of Schr\"odinger operators on finite regular metric trees through a relation to orthogonal polynomials that provides a graphical perspective.  As the Robin vertex parameter tends to $-\infty$, a narrow cluster of finitely many eigenvalues tends to $-\infty$, while the eigenvalues above this cluster remain bounded from below.  Certain ``rogue" eigenvalues break away from this cluster and tend even faster toward~$-\infty$.
The spectrum can be visualized as the intersection points of two objects in the plane---a spiral curve depending on the Schr\"odinger potential, and a set of curves depending on the branching factor, the diameter of the tree, and the Robin parameter.
}}

\vspace{3ex}
\noindent
\begin{mbox}
{\bf Key words:}  quantum graph, Jacobi operators on trees, Robin spectrum, orthogonal polynomials
\end{mbox}

\vspace{3ex}

\noindent
\begin{mbox}
{\bf MSC:}  34B45, 34B08, 34A26, 34L15, 34L20, 34L05
\end{mbox}
\vspace{3ex}


\hrule
\vspace{1.1ex}

\section{Introduction}\label{sec:trees}

A quantum tree is a Schr\"odinger operator on a metric tree graph.  On each edge, the operator acts as $-d^2/dx^2+q(x)$, and the edges are coupled by self-adjoint vertex conditions.  
Regular tree operators have large symmetry groups, which allows in-depth analyses by reduction to simpler linear graphs.  Much is known about their spectra, including distribution of eigenvalues for finite trees and, for infinite trees, band-gap structure and infinite-multiplicity eigenvalues, particularly high-energy asymptotics.  The present work emphasizes analysis of low-energy spectral phenomena, particularly certain ``rogue" eigenvalues that tend to $-\infty$ as the vertex condition becomes extreme.

The vertex condition relates the value of a function at any vertex to its flux into the adjacent edges, and the proportionality constant is called the Robin parameter $\alpha$, defined precisely below.
\rev{The condition can be thought of as a singular $\delta$-potential of strength $\alpha$ at the vertex; see, for example~\cite[Eq.\,2.6]{ExnerPost2009}; and it is called the $\delta$~vertex condition.  It is also called the Robin vertex condition~\cite[Thm.\,2.2]{BerkolaikoKuchment2012a}, and it generalizes the usual Robin, or transmission, boundary condition in linear ODE theory.  
It is proved in \cite[\S2--3]{ExnerPost2009} that a quantum graph with the Robin vertex condition can be realized as the limit of a Schr\"odinger operator on a thickened graph, where the scaling of thickened edges and vertices is taken in an appropriate way.
Different limits of thickened graphs lead to other types of self-adjoint vertex conditions; see also \cite{ExnerPost2005,CheonExnerTurek2010a} and references therein for analyses and discussions of these limits as well as their physical applications.}

\rev{
These analyses of
\cite{ExnerPost2005,ExnerPost2009,CheonExnerTurek2010a}  
strengthen the foundations of quantum graphs as models of real physical systems, especially with regard to the different vertex conditions that render them self-adjoint.  
Quantum graphs model a rich variety of physical systems, and we refer the reader to such works as \cite{AlbeverioGesztesyHoegh-Kroh1988a,Kuchment2004a}, the collection~\cite{ExnerKeatingKuchment2008}, and the monograph~\cite{BerkolaikoKuchment2013}.
The dependence of spectral properties on the Robin parameter, particularly analyticity and eigenvalue interlacing, has been investigated in~\cite{BerkolaikoKuchment2012a}.
}

The present work offers a graphical approach for analysis of the spectrum of finite regular rooted quantum trees with symmetric potential $q(x)$, which is based on orthogonal polynomials determined by the graph and how they interact with the Schr\"odinger operator.  This approach (1) offers a visualization of quantitative details of the spectrum, including clusters of eigenvalues and intermediate eigenvalues that interlace the clusters (the clusters become bands for infinite trees); (2) clarifies the role of the potential of the operator versus the role of the branching number, the length of the graph, and the vertex condition; and (3) allows one to prove new results on rogue eigenvalues that escape to $-\infty$ as the Robin parameter $\alpha$ in the vertex condition becomes large and negative.  These curious eigenvalues include an exponentially narrow cluster plus extra eigenvalues that break free below the clusters.  The rest of the eigenvalues have a lower bound independent of the Robin parameter.  In addition, (4) the relation between the spectra of large-diameter finite graphs and infinite graphs is clarified.

\rev{
The intuitive reason for the rogue eigenvalues is that a negative Robin parameter $\alpha$ imparts a negative potential localized at the vertices.  This results in a finite number of eigenvalues at negative energies for the finite tree.  These gather in thin energy clusters that tend to $-\infty$ as $\alpha\to-\infty$.  What was unexpected is that some of the lowest eigenvalues break off from the cluster and tend even faster to $-\infty$.  The graphical apparatus we introduce demonstrates this phenomenon in a clear visual way and leads to a proof of the rogue eigenvalues' asymptotics.}

\rev{
The reduction by symmetry mentioned at the beginning comes from the invariance of the tree graphs under cyclic permutations of the subtrees below any vertex~\cite{Carlson2000,NaimarkSolomyak2000,Solomyak2003,Solomyak2004}.  It is sometimes called ``radial symmetry".  It results in a reduction into energy-dependent Jacobi matrix operators.
This symmetry still permits the edge lengths, potentials, branching number, and Robin parameter to vary with the generation, or the distance from a root vertex.  
In the present paper, the constancy of these four quantities always leads to solely absolutely continuous spectrum as the length of the graph becomes infinite and the clusters turn into spectral bands. 
 But absolutely continuous spectrum for infinite trees is a rare event under radial symmetry alone~\cite{BreuerFrank2009a}. 
Certain sparsity conditions on the edges, as they depend on the generation, lead to singular continuous and dense point spectrum \cite{Breuer2007a,Breuer2007b,ExnerLipovsky2010a}. 
 Absence of singular continuous spectrum is known for periodic Jacobi matrices on (discrete) trees~\cite{AvniBreuerSimon2020a}, and point spectrum is possible if the branching number can vary~\cite{Aomoto1991a}. 
Interestingly, the form of the reduction into linear (one-dimensional) graph operators is valid for a much broader class of symmetries possessed by ``path-commuting" graphs~\cite{BreuerKeller2013,BreuerLevi2020}}. 

The main results of the present work are Theorems~\ref{thm:spectrum1} and~\ref{thm:spectrum2} in section~\ref{sec:linearspectrum} on the spectra of the finite linear-graph reductions, which are components of the spectrum of the full tree, and these spectra are illustrated in Fig.~\ref{fig:Alphas}, \ref{fig:SpecB}, and~\ref{fig:SpecBo}.  In addition to the new results discussed above, the theorems reproduce several known results, such as a Weyl law for the asymptotic distribution of eigenvalues proved in~\cite[\S5.2]{Carlson2000} and \cite[\S5.2]{Solomyak2003}.
For a survey of regular quantum graphs up to 2004, see~\cite{Solomyak2004}.


\smallskip

{\bfseries Notation.}  The underlying graph $\Gamma_{\!n}$ has a root vertex, $n$ branching levels, and a common branching factor $b$, as illustrated in Fig.\,\ref{fig:tree}.  
To make descriptions concrete, place the root vertex at the top and identify each edge with the $x$-interval $[0,1]$ directed downward.  Each vertex except the $b^n$ terminal ones at the bottom has $b$ outgoing edges emanating downward, and each vertex except the root has one incoming edge entering it from above.  Vertices besides the root and terminal ones will be called interior vertices.

The Schr\"odinger operator $A_n$ takes the form $-d^2/dx^2+q(x)$ with the same electric potential $q\in L^2[0,1]$ on each edge.  We take the potential to be symmetric about the midpoint of the edge, that is, $q(1-x) = q(x)$.  The domain of definition of $A$ consists of continuous functions on $\Gamma_{\!n}$ that, together with their derivatives along the edges, are square integrable, and are subject to a Robin vertex condition, which requires that the sum of all outgoing derivatives from a vertex be equal to $\alpha$ times the value at the vertex.  The real number $\alpha$ is the {\bfseries\em Robin parameter.}  More precisely:  Let $u$ be a function on $\Gamma_{\!n}$ with value $u_*$ at vertex $v$.  If $v$ is interior, denote $u$ by $u_0(x)$ on the incoming edge and by $u_i(x)$, $1\leq i\leq b$ on the outgoing edges.  The vertex condition~is
\begin{equation}\label{Robin}
  \sum_{i=1}^bu_i'(0) - u_0'(1) \;=\; \alpha\, u_*.
\end{equation}
At the root vertex, the condition is
$\sum_{i=1}^bu_i'(0) \!=\! \alpha\, u_*$, and at the terminal vertices, it is $-u_0'(1) \!=\! \alpha\, u_*$.
This condition makes $A_n$ a self-adjoint operator in $L^2(\Gamma)$~\cite{AlbeverioGesztesyHoegh-Kroh1988a,Exner1997,BerkolaikoKuchment2012a}.  The pair $(\Gamma_{\!n},A_n)$ is thus a quantum tree.

\begin{figure}[h]
\centerline{\includegraphics[width=0.45\linewidth]{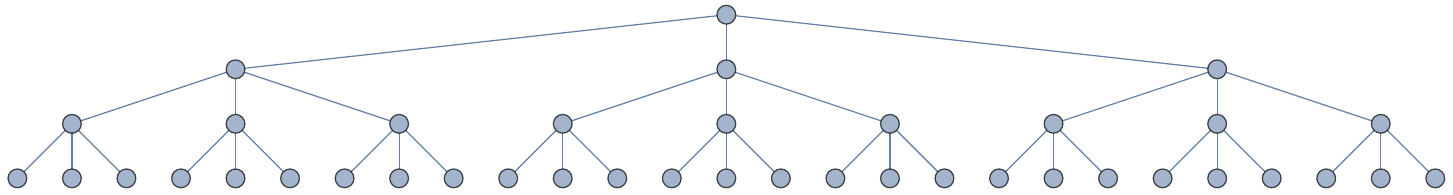}\vspace{-1.5ex}}
 \caption{\small A regular tree graph $\Gamma_{\!n}$ with $n\!=\!3$ branching levels and branching factor $b=3$.}
 \label{fig:tree}
\end{figure}

A quantum tree $(\Gamma_{\!\infty},A_\infty)$ is obtained when the number of branching levels is infinite and there are no terminal vertices.  This is nearly the infinite regular tree studied by Carlson~\cite{Carlson1997}, except that one of the vertices, the root, has degree $b$ instead of $b+1$.  The infinite regular tree has a sequence of spectral bands with one eigenvalue of infinite multiplicity in each gap~\cite{Carlson1997}.
  The root of $(\Gamma_{\!\infty},A_\infty)$ may be considered to be a defect of an infinite regular tree, and it produces additional eigenvalues in the gaps.  A variation of this is an infinite rooted regular tree with a potential that decays with the distance from the root, which also produces additional eigenvalues in the gaps~\cite{SobolevSolomyak2002}.
  The spectrum of our infinite tree with either Robin or Dirichlet conditions at the root is discussed briefly in Section~\ref{sec:infinite}, although our focus is the finite rooted tree.

\smallskip
{\bfseries Organization.} This article is organized so as to convey the results as quickly as possible, supported by graphical representations of the spectra.  This requires first the reduction to linear discrete graph operators in Section~\ref{sec:linear} and a description of the graphical analysis in Section~\ref{sec:graphical}.  The main theorems are in Section~\ref{sec:linearspectrum}.  The orthogonal polynomials are introduced and analyzed in Section~\ref{sec:orthogonal} and used to prove the theorems in Section~\ref{sec:proofs}.  Section~\ref{sec:infinite} offers a short discussion of infinite trees, and Section~\ref{sec:appendixmoments} is an appendix on the moments and the weight associated with the orthogonal polynomials and their bearing on the spectra of quantum trees.

\begin{figure}[ht]
\centerline{\includegraphics[width=0.99\linewidth]{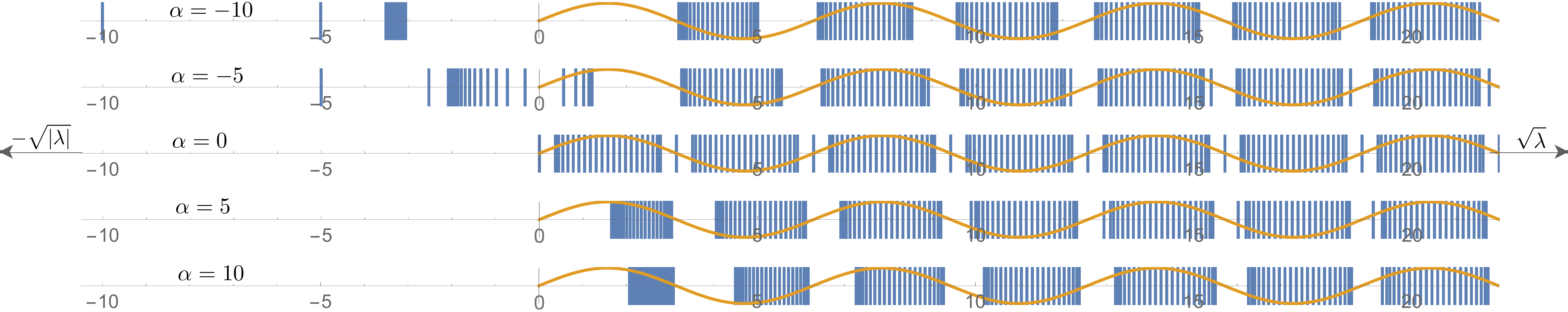}\vspace{-1.5ex}}
 \caption{\small  The spectrum of the linear quantum graph $(\Lambda_n,B_n)$ as a function of $\mu\!=\!\mathrm{sgn}\,\lambda\sqrt{|\lambda|}$ ({\itshape i.e.,} $\lambda\!=\!|\mu|\mu$) for various values of the Robin parameter~$\alpha$.  The potential is $q(x)\!=\!0$, and the branching number is $b\!=\!2$, and the length of the graph is $n\!=\!22$.  Clusters of eigenvalues are confined between the Dirichlet eigenvalues of the interval [0,1], which are the non-zero roots of $\sin\mu$, graphed in orange.  For $\alpha\!=\!0$, the spectrum is periodic for $\mu\!>\!0$.  For any $\alpha$, the spectrum approaches that for $\alpha\!=\!0$ as $\mu\to\infty$.  For $\alpha\!>\!0$, eigenvalues are pushed to the right, and for $\alpha\!<\!0$, they are pushed to the left.  As $\alpha\!\to\!-\infty$, the leftmost cluster moves to $-\infty$ and two rogue eigenvalues break off.}
 \label{fig:Alphas}
\end{figure}

\begin{figure}[h]
\centerline{\includegraphics[width=0.95\linewidth]{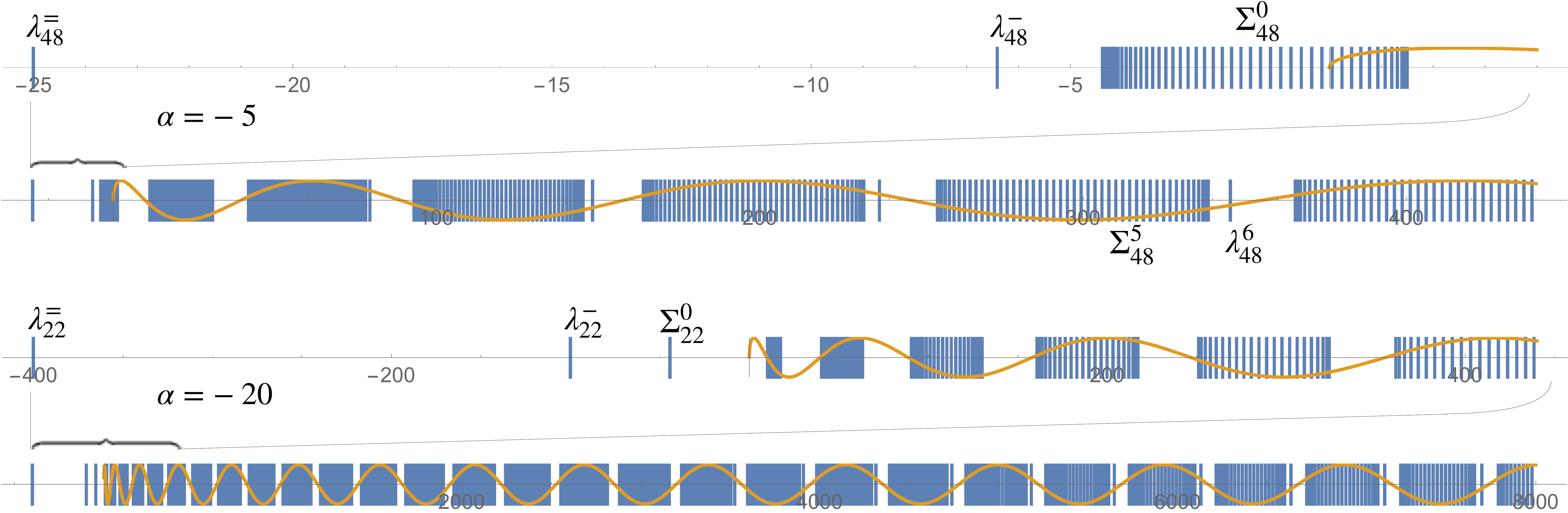}\vspace{-1.5ex}}
 \caption{\small The spectrum of the linear graph $(\Lambda_n,B_n)$, with the Robin condition at all vertices, is part of the spectrum of the regular quantum tree $(\Gamma_{\!n},A_n)$.  Here, $b=2$, the Robin parameter is $\alpha\!=\!-5$ and $\alpha\!=\!-20$, and the potential is $q(x)\!=\!0$.  (The subscript $n\!=\!48$ or $n\!=\!22$ indicates the number of levels of the tree.)  As $\alpha\to-\infty$, two {\em rogue eigenvalues} $\lambda_n^=$ and $\lambda_n^-$ and the {\em rogue cluster} $\Sigma_n^0$ tend to $-\infty$ as $-\alpha^2$, $-b^{-2}\alpha^2$, and $-(b+1)^{-2}\alpha^2$, with $\Sigma_n^0$ becoming exponentially thin.  Each {\em cluster} $\Sigma_n^k$ for $k\geq1$ lies between the Dirichlet eigenvalues $\lambda_D^k$ and $\lambda_D^{k+1}$ of the interval $[0,1]$ (which for $q(x)\!=\!0$ are the non-zero roots of $\sin(\sqrt{\lambda})$, graphed in orange).  Each two successive clusters are interlaced by {\em an intermediate eigenvalue}~$\lambda_n^k$.  For large negative $\alpha$, the intermediate eigenvalue $\lambda_n^k$ is very close to the upper end of the cluster $\Sigma_n^{k-1}$ for smaller $\lambda$ and moves rightward for larger~$\lambda$.}
 \label{fig:SpecB}
\end{figure}

\begin{figure}[ht]
\centerline{\includegraphics[width=0.95\linewidth]{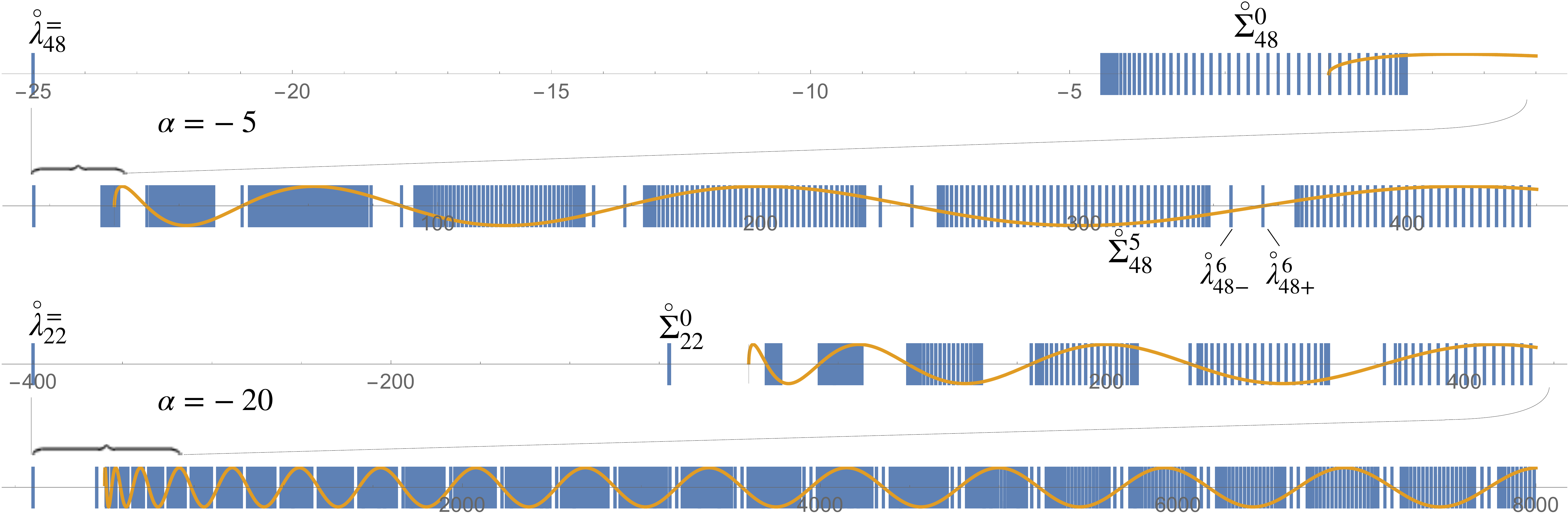}\vspace{-1.5ex}}
 \caption{\small The spectrum of the linear graphs $(\Lambda_m,\Bo_m)$ for $m\leq n$ with the Robin condition at all vertices except the root, is part of the spectrum of the regular quantum tree $(\Gamma_{\!n},A_n)$.  Here, $b=2$, the Robin parameter is $\alpha\!=\!-5$ and $\alpha\!=\!-20$, and the potential is $q(x)\!=\!0$.  (The subscript $m\!=\!49$ or $m\!=\!21$ indicates the number of levels of the tree.)  As $\alpha\to-\infty$, one {\em rogue eigenvalue} $\lo_m^=$ and the {\em rogue cluster} $\So_m^0$ tend to $-\infty$ as $-\alpha^2$ and $-(b+1)^{-2}\alpha^2$, with $\Sigma_n^0$ becoming exponentially thin.  Each {\em cluster} $\So_m^k$ for $k\geq1$ lies between the Dirichlet eigenvalues $\lambda_D^k$ and $\lambda_D^{k+1}$ of the interval (which for $q(x)\!=\!0$ are the non-zero roots of $\sin(\sqrt{\lambda})$, graphed in orange).  Each two successive clusters are interlaced by {\em two intermediate eigenvalues}~$\lo_{m-}^k$ and $\lo_{m+}^k$.  For large negative $\alpha$, the intermediate eigenvalue $\lo_{m-}^k$ is very close to the upper end of cluster $\So_m^{k-1}$ for smaller $\lambda$ and moves rightward for larger~$\lambda$.}
 \label{fig:SpecBo}
\end{figure}

\section{Reduction to linear discrete graphs}\label{sec:linear}

A regular homogeneous quantum graph has a large group of symmetries, which reduces it to a direct sum of linear graphs.  This has been known since the work of Carlson~\cite{Carlson2000} and Naimark and Solomyak~\cite{NaimarkSolomyak2000,Solomyak2003,Solomyak2004}.

The decomposition described in~\cite[Theorem~7.2]{NaimarkSolomyak2000}, \cite[\S4]{Solomyak2003}, and~\cite[Theorem~3.2]{Solomyak2004} proceeds as follows.  Let $(\Gamma_{\!n},\Asym_n)$ be the quantum tree equal to $(\Gamma_{\!n},A_n)$ but restricted to those functions in the domain of $A_n$ that are constant across all $b^\ell$ points at any given level $\ell$ of $\G_n$.  Such functions $f$ are called {\em completely symmetric}.
Let $(\Gamma_{\!n},\Ao_n)$ be the same quantum tree as $(\Gamma_{\!n},A_n)$ except with the Dirichlet (zero-value) instead of Robin condition imposed at the root vertex; and let $(\Gamma_{\!n},\Aosym_n)$ be its restriction to completely symmetric functions.
Then
\begin{equation}
  A_n \;\cong\; \Asym_n
      \;\oplus\; \sum_{m=1}^n \big(\Aosym_m\big)^{(b-1)b^{n-m}},
\end{equation}
in which the power refers to the $(b-1)b^{n-m}$-fold direct sum of $\Aosym_m$ with itself.

The next step in the simplification of $(\G_n,A_n)$ is the reduction of $\Asym_n$ and $\Aosym_n$ to Schr\"odinger operators on a linear graph $\Lambda_n$ with $n+1$ vertices and $n$ edges, as depicted in Fig.\,\ref{fig:linear}.  This process is described in~\cite{Carlson2000} and in~\cite[\S4,\,eqn.~6]{Solomyak2003}.  Essentially, since a function in $\dom(\Asym_n)$ or $\dom(\Aosym_n)$ is invariant across any given level, it is determined by its value on a single edge of each level, so that all $b^{\ell+1}$ edges emanating down from vertices at level $\ell$ can be compressed into just one edge.
A function $u:\Lambda_n\to\CC$ is denoted by the $n$-tuple $u=\left\{ u_k \right\}_{k=1}^n$ of restrictions to the edges of $\Lambda_n$, in which each $u_k$ is a function of the unit interval $[0,1]$ that parameterizes the edge.

\begin{figure}[ht]
\centerline{\includegraphics[width=0.68\linewidth]{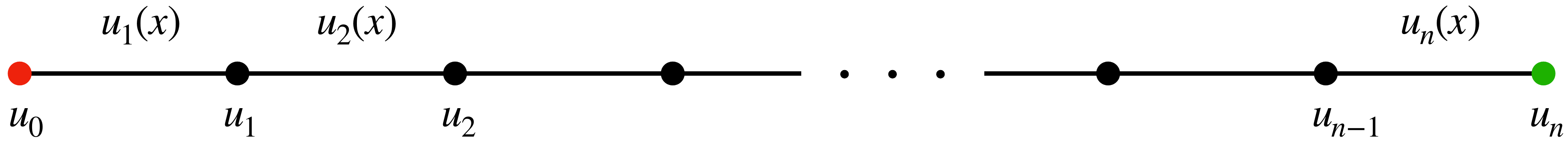}\vspace{-.5ex}}
 \caption{\small The notation for functions $\{u_k(x)\}_{k=1}^n$ on the linear graphs $\Lambda_n$ and their values $u_k$ at the vertices.}
 \label{fig:linear}
\end{figure}

The Schr\"odinger operators $(\Lambda_n,B_n)$ and $(\Lambda_n,\Bo_n)$ on these linear quantum trees are self-adjoint with respect to a weighted inner product,
\begin{equation}\label{inner}
  \langle f, g \rangle = \sum_{k=1}^n \, b^k\! \int_0^1 f_k(x) \,\bar g_k(x) \,dx\,.
\end{equation}
The Robin condition for $B_n$ induced by the vertex condition (\ref{Robin}) is
\begin{equation}\label{Robin1}
  - u'_i(1) + b u'_{i+1}(0) - \alpha u_{i+1}(0) = 0, \quad \text{for } 1\leq i\leq n-1 \\
\end{equation}
for the interior vertices, with $bu'_1(0) - \alpha f_1(0) = 0$ at the root and $-u'_n(1) - \alpha u_n(1) = 0$ at the end.  The condition for $\Bo_n$ is the same except for the boundary condition at the root vertex, which is $u_1(0) = 0$.
The operators $\Asym_n$ and $B_n$ are unitarily similar to each other, as are the operators $\Aosym_n$ and $\Bo_n$.  The decomposition theorem becomes
\begin{equation}\label{decomposition}
  A_n \;\cong\; B_n \;\oplus\;
      \bigoplus_{m=1}^n \big(\Bo_m\big)^{(b-1)b^{n-m}}.
\end{equation}
The spectrum of $A_n$ can now be computed just by computing the spectrum of two types of linear quantum trees, $B_n$ and $\Bo_m$.

\begin{theorem}[\cite{NaimarkSolomyak2000,Solomyak2004}]\label{thm:decomposition}
Let $(\Gamma_{\!n},A_n)$ be a regular quantum tree with vertex conditions (\ref{Robin}), as described above, and let $B$ and $\Bo$ be the associated linear quantum graphs with vertex conditions described in (\ref{Robin1}) and the following text.  Then the spectrum of $A_n$ is equal to the following union of spectra:
\begin{equation}\label{union}
  \sigma(A_n) \;=\; \sigma(B_n) \;\cup\;
    \bigcup_{m=1}^n \sigma(\Bo_m).
\end{equation}
Each of these spectra consists only of a discrete set of eigenvalues, and the elements of $\sigma(\Bo_m)$ have multiplicity $(b-1)b^{n-m}$ for $A_n$, provided all these $n+1$ spectral sets are mutually disjoint.
\end{theorem}

The equations $Bu=\lambda u$ and $\Bo u=\lambda u$ can be reduced to problems on discrete graphs since the solutions to $(-d^2/dx^2+q(x)-\lambda)u=0$ on $[0,1]$ are determined by their values at $0$ and $1$, unless $\lambda$ is a Dirichlet eigenvalue.  This reduction is a standard procedure in the study of quantum graphs.  Let $c(x,\lambda)$ and $s(x,\lambda)$ be solutions of \,$-u'' + q(x) u = \lambda u$\, satisfying the initial conditions
\begin{eqnarray*}
  c(0,\lambda) = 1 & s(0,\lambda)=0 \\
  c'(0,\lambda) = 0 & s'(0,\lambda)=1,
\end{eqnarray*}
in which the prime denotes the derivative with respect to the first argument $x$.
Under the symmetry assumption $q(1-x)=q(x)$, one has $c(1,\lambda)=s'(1,\lambda)$.
Define
\begin{eqnarray}
  c(\lambda) &:=& c(1,\lambda) = s'(1,\lambda) \\
  s(\lambda) &:=& s(1,\lambda).
\end{eqnarray}
As long as $s(\lambda)\not=0$ ($\lambda$ is not a Dirichlet eigenvalue of $-d^2/dx^2 + q(x)$ on $[0,1]$), the derivatives of $u$ at the endpoints are obtained from the values of $u$ by the Dirichlet-to-Neumann map
\begin{equation}\label{DtN}
\renewcommand{\arraystretch}{1.2}
\frac{1}{s(\lambda)}
\left[\hspace{-5pt}
\begin{array}{cc}
  -c(\lambda) & 1 \\ 1 & -c(\lambda)
\end{array}
\hspace{-3pt}\right]
\renewcommand{\arraystretch}{1.2}
\left[\hspace{-5pt}
\begin{array}{c}
  u(0) \\ u(1)
\end{array}
\hspace{-5pt}\right]
 =
\renewcommand{\arraystretch}{1.2}
\left[\hspace{-5pt}
\begin{array}{c}
  u'(0) \\ -u'(1)
\end{array}
\hspace{-5pt}\right].
\end{equation}
We denote the Dirichlet spectrum for the potential $q(x)$ by
\begin{equation}\label{sigmaD}
  \sigma_{\!D}(q) \;:=\;
  \{  \lambda\in\RR : s(\lambda) = 0 \}.
\end{equation}

Let $u=\left\{ u_k(x) \right\}_{k=1}^n$ satisfy $Bu=\lambda u$, and let $u_k$ be the values at the vertices as in Fig.\,\ref{fig:linear}:
$u_0 := u_1(0)$ and $u_k := u_{k+1}(0) = u_k(1)$ and $u_n := u_n(1)$.
Since $-u''_k + q(x) u_k = \lambda u_k$, the Robin condition (\ref{Robin1}) can be written solely in terms of the~$u_k$ if $s(\lambda)\not=0$,
\begin{equation}\label{Robindiscrete}
\begin{aligned}
-(cb+s\alpha)u_0+bu_1&=0\\
u_{k-1}-(c(b+1)+s\alpha)u_k+bu_{k+1}&=0 \quad (0< k<n) \\
u_{n-1}-(c+s\alpha)u_n&=0.
\end{aligned}
\end{equation}
This is a homogeneous system for $\{u_k\}_{k=0}^n$ with $(n+1)\!\times\!(n+1)$ matrix of coefficients 
\begin{equation}\label{Mn}
M_n\;=\;
{\small
\begin{bmatrix}
cb+s\alpha &-b &0 &0  & \cdots &0\\
-1 &c(b+1)+s\alpha &-b  &0 & \cdots &0\\
0 &\ddots &\ddots &\ddots  &\ddots &\vdots\\
\vdots &\ddots &\ddots &\ddots  &\ddots &0\\
0 &\cdots &0 & -1 &c(b+1)+s\alpha &-b\\
0 & \cdots &0  &0 &-1 &c+s\alpha         
\end{bmatrix},
}
\end{equation}
which realizes the reduction to a Jacobi matrix (with respect to the inner product induced by the matrix $\mathrm{diag}(1,b,...,b^n)$) mentioned in the Introduction, with the $\lambda$ dependence occurring in $c(\lambda)$ and $s(\lambda)$.
Define the determinants
\begin{equation}\label{Dn}
   D_n = \det M_n \quad (n\geq1)\,;
   \qquad
   D_0 = s\alpha\,,
   \quad
   D_{-1} = 1-c^2\,.
\end{equation}
The values of $\lambda\not\in\sigma_{\!D}(q)$ for which $D_n=0$ are the eigenvalues of the $n$-level quantum tree.  The case that $\lambda\in\sigma_{\!D}(q)$ must be handled separately.
For the case of the Dirichlet condition at the root vertex (but retaining Robin conditions at the other vertices), just the first row of the matrix of coefficients differs from the Robin case.  The new matrix is
\begin{equation}\label{Mno}
\Mo_n\;=\;
{\small
\begin{bmatrix}
1 &0 &0 &0  &\cdots &0\\
-1 &c(b+1)+s\alpha &-b  &0 &\cdots &0\\
0 &\ddots &\ddots &\ddots  &\ddots &\vdots\\
\vdots &\ddots &\ddots &\ddots  &\ddots &0\\
0 &\cdots &0 & -1 &c(b+1)+s\alpha &-b\\
0 & \cdots &0  &0 &-1 &c+s\alpha         
\end{bmatrix}.
}
\end{equation}
Define the determinants
\begin{equation}\label{Don}
   \Do_n = \det \Mo_n \quad (n\geq1)\,;
   \qquad
   \Do_0 = 1\,,
   \quad
   \Do_{-1} = c\,.
\end{equation}
Values of $\lambda\not\in\sigma_{\!D}(q)$ for which $\Do_n=0$ are the eigenvalues of the $n$-level quantum tree with the Dirichlet condition at the root.

\begin{theorem}\label{thm:recurrence}  
The quantities $D_n$ and $\Do_n$ satisfy the same second-order recurrence relation,
\begin{equation}\label{Drecurrence}
  D_n \;=\; [c(b+1)+s\alpha]\, D_{n-1} - b\, D_{n-2}
  \qquad
  (n\geq1),
\end{equation}
\begin{equation}\label{Dorecurrence}
  \Do_n \;=\; [c(b+1)+s\alpha]\, \Do_{n-1} - b\, \Do_{n-2}
  \qquad
  (n\geq1).
\end{equation}
\end{theorem}

\begin{proof} 
  Define $N_n$ to be the $(n+1)\!\times\!(n+1)$ matrix whose entries are identical to those of $M_n$ except that the $(n+1,n+1)$ (lower right) entry is $c(b+1)+s\alpha$ in place of $c+s\alpha$.
Define the determinants
$E_n = \det N_n$ for $n\geq0$ and $E_{-1} = 1$ and $E_{-2} = cb^{-1}$.
One computes that $D_n = E_n - cb\, E_{n-1}\;(n\geq-1)$
and that $\{E_n\}$ satisfies the recursion
\begin{equation}
  E_n = [c(b+1)+s\alpha] E_{n-1} - b E_{n-2}
  \qquad
  (n\geq0),
\end{equation}
and therefore so does $\{E_{n-1}\}\;(n\geq1)$.
$\{D_n\}$ being a linear combination of these two also satisfies this recursion.

For~$\{\Do_n\}$, define the matrices $\No_n$ by replacing the $(n+1,n+1)$ (lower right) entry of $\Mo_n$ by $c(b+1)+s\alpha$, and define
$E_n = \det N_n$ for $n\geq0$ and $E_{-1} = 0$ and $E_{-2} = -b^{-1}$.
Again, one computes the relations $\Do_n = \Eo_n - cb\, \Eo_{n-1}$ ($n\geq-1$) and
$\Eo_n = [c(b+1)+s\alpha] \Eo_{n-1} - b \Eo_{n-2}$ ($n\geq0$)
and thence the desired recurrence for $\{\Do_n\}$.
\end{proof}

\section{Graphical analysis}\label{sec:graphical}

By replacing the spectral functions $c(\lambda)$ and $s(\lambda)$ in $D_n$ and $\Do_n$ by independent variables $y$ and $z$, the functions $D_n=D_n(y,z)$ and $\Do_n=\Do_n(y,z)$ become polynomials in two variables.  An energy $\lambda\not\in\sigma_{\!D}(q)$ is an eigenvalue of the linear quantum graph $(\Lambda_n,B_n)$ whenever
\begin{equation}
  D_n\big(c(\lambda),s(\lambda)\big) = 0
\end{equation}
and is an eigenvalue of $(\Lambda_n,\Bo_n)$ whenever $\Do_n\big(c(\lambda),s(\lambda)\big) = 0$.
This point of view has the advantage that the spectrum is represented geometrically by the intersection points of two objects in the $yz$-plane that separate the role of the potential $q(x)$ from the role of the parameters $\alpha$, $b$, and $n$. 
These objects are

\begin{enumerate}
  \item  The curve $\Spiral$ in the $yz$-plane parameterized by $y\!=\!c(\lambda)$ and $z\!=\!s(\lambda)$ ($\lambda\in\RR$).  This curve is determined by the potential $q(x)$.
  \item  The zero sets $D_n(y,z)=0$ and $\Do_n(y,z)=0$.  These are determined by the Robin parameter $\alpha$, the splitting degree $b$ of the tree, and the number of levels $n$.  Both $\alpha$ and $b$ appear in the recurrence relation given by Theorem~\ref{thm:recurrence}, and $\alpha$ appears in the initial condition (\ref{Dn}) for $\{D_n\}$.
\end{enumerate}

\noindent
\rev{The superposition of these two objects is shown in Fig.\,\ref{fig:DandSq0} for $q(x)\!=\!0$ and in detail in Fig.\,\ref{fig:DandS} for a nonzero potential $q(x)$.  Each point of intersection corresponding to $\lambda\not\in\sigma_{\!D}(q)$ gives an eigenvalue.  
The rough qualitative dependence of the spectrum on $\alpha$, as shown in~Fig.~\ref{fig:Alphas}, can be deduced readily from Fig.\,\ref{fig:DandSq0}.}

Note that the elements of $\sigma_{\!D}(q)$ are the values of $\lambda$ where $\Spiral$ intersects the $y$-axis $z\!=\!0$, that is, where $s(\lambda)=0$, and whether these are eigenvalues of the tree operator is treated in Sec.~\ref{intermediate}.

\begin{figure}[ht]
\centerline{\includegraphics[width=0.95\linewidth]{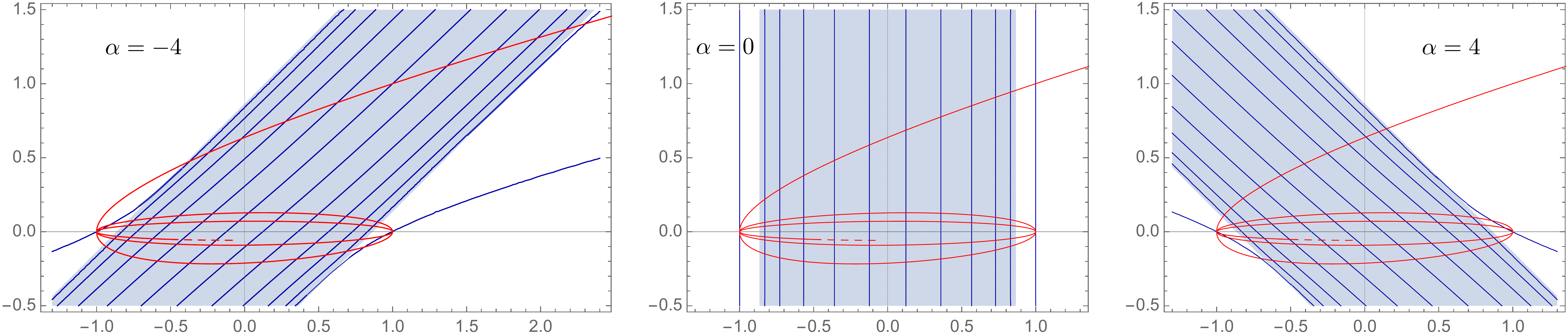}}
 \caption{\small The blue curves form the zero-sets of $D_n(y,z)$, and the red curve is $\Spiral: (y=\cos(\sqrt{\lambda}),\, z=\mathrm{sinc}(\sqrt{\lambda}))$ for potential $q(x)\!=\!0$.
 The $\lambda$-values of their intersection points are the eigenvalues of the linear quantum graph $(\Lambda_n,B_n)$.
The branching number is $b\!=\!3$, the Robin parameter is $\alpha\!=\!-2$, and the length of the graph is $n\!=\!11$.  The shaded part is the ``oscillatory region".  Compare Fig.~\ref{fig:Alphas}, which shows the eigenvalues for various values of~$\alpha$.
 }
 \label{fig:DandSq0}
\end{figure}

\begin{figure}[ht]
\centerline{\includegraphics[width=0.9\linewidth]{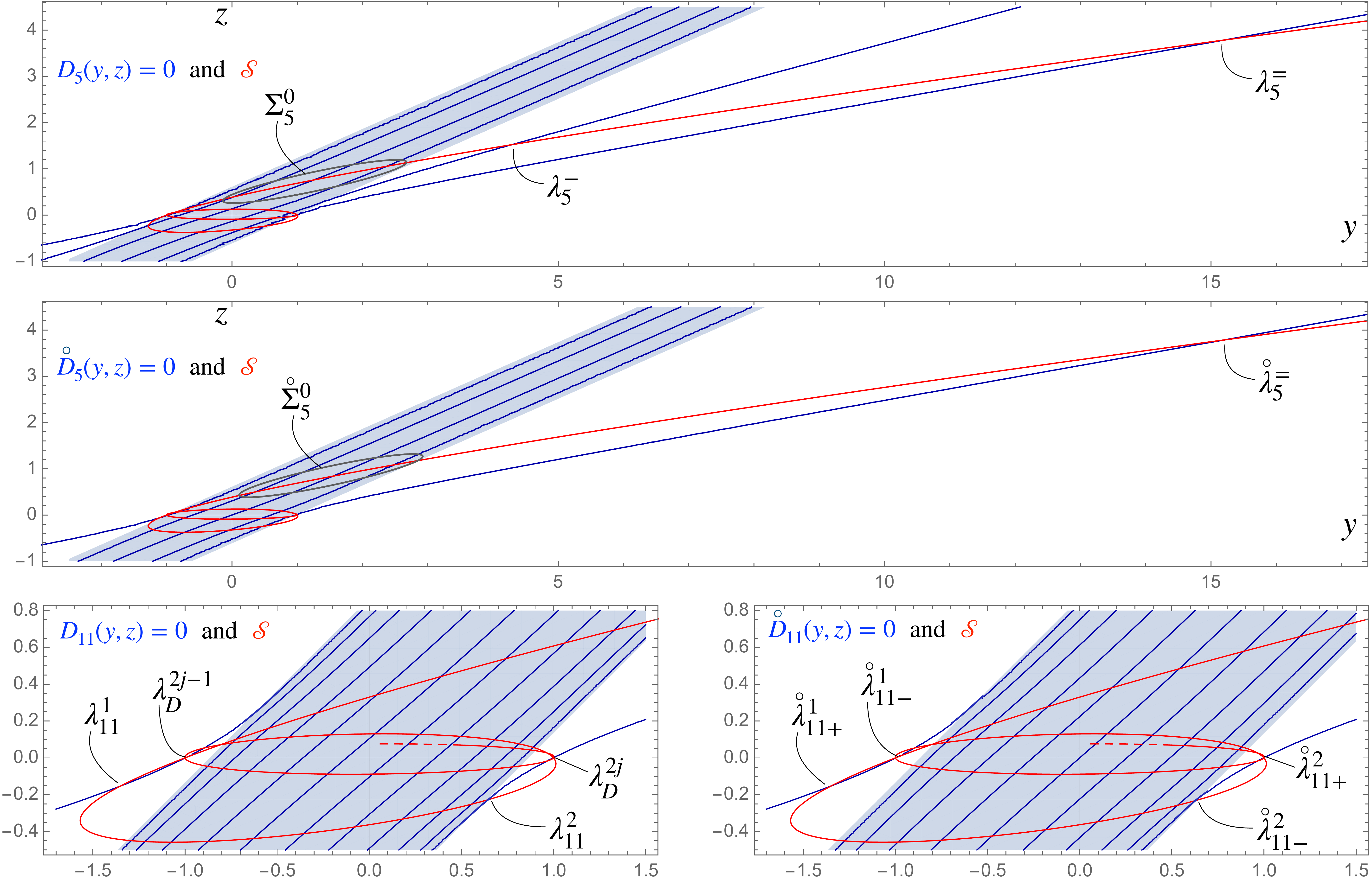}}
 \caption{\small The blue curves form the zero-sets of $D_n(y,z)$ or $\Do_n(y,z)$, which depend on the Robin parameter $\alpha$, the branching factor $b$, and the number of levels $n$ of the tree; and the red curve is $\Spiral: (y=c(\lambda),\, z=s(\lambda))$, which depends on $q(x)$ only.  The $\lambda$-values of their intersection points are the eigenvalues of the linear quantum graphs $(\Lambda_n,B_n)$ and $(\Lambda_n,\Bo_n)$.  The shaded part is the ``oscillatory region".
 {\bfseries Top:}  The top two plots show the rogue eigenvalues $\lambda_n^-$, $\lambda_n^=$, and $\lo_n^=$ and the rogue clusters $\Sigma_n^0$ and $\So_n^0$, which tend to $-\infty$ as $\alpha\to-\infty$.  Here, $\alpha=-4$ and $b=1.5$ for illustration, and $q(x)\!=\!-16\,\cchi_{[1\!/\!3,2\!/\!3]}(x)$.
{\bfseries Bottom:} For $(\Lambda_n,B_n)$, an intermediate eigenvalue $\lambda_n^k$ and the Dirichlet eigenvalue $\lambda_D^k$ for the interval interlace the bands $\Sigma_n^k$ for $k\geq1$; and for $(\Lambda_n,\Bo_n)$, two intermediate eigenvalues and the $k^\text{th}$ Dirichlet eigenvalue $\lo_{n-}^k<\lambda_D^k<\lo_{n+}^k$ interlace the bands $\So_n^k$ for $k\geq1$ (one of $\lo_{n\pm}^k$ is very close to $\lambda_D^k$).  Here, $\alpha=-4$ and $b=3$, and $q(x)\!=\!-22\,\cchi_{[1\!/\!3,2\!/\!3]}(x)$.
 }
 \label{fig:DandS}
\end{figure}

\section{Spectrum of linear and tree graphs}\label{sec:linearspectrum}

The two theorems in this section give detailed accounts of the spectra of the linear quantum graphs with Robin and Dirichlet conditions at the root vertex, with particular attention to the behavior as the Robin parameter $\alpha$ tends to $-\infty$.  The qualitative structure of the spectra comes from the graphical analysis of the previous section, and proofs are given later.  The spectrum of our regular quantum tree is the union of spectra of linear graphs, as stated in Theorem~\ref{thm:decomposition}.

The eigenfunctions are also discussed.

\subsection{Eigenvalues}\label{sec:eigenvalues}

The Dirichlet eigenvalues of $-d^2/dx^2 + q(x)$ on $[0,1]$ separate clusters of eigenvalues of the linear quantum graphs $(\Lambda_n,B_n)$ and $(\Lambda_n,\Bo_n)$.  
\rev{This can be deduced from results on eigenvalue interlacing for quantum graphs~\cite[\S5]{BerkolaikoKuchment2012a}, where $\alpha\!=\!\infty$ corresponds to Dirichlet vertex conditions.  The graphical analysis illustrates this in a new light---see Fig.\,\ref{fig:DandS}.}  Denote the these eigenvalues by
\begin{equation}
  \lambda_D^1 < \lambda_D^2 < \dots < \lambda_D^k < \dots,
\end{equation}
and note that $\sigma_{\!D}=\sigma_{\!D}(q)=\{\lambda_D^k\}_{k=1}^\infty$ (recall eq.\,\ref{sigmaD}).
The eigenvalues of $(\Lambda_n,B_n)$ or $(\Lo_n,B_n)$ that lie below the lowest Dirichlet eigenvalue $\lambda_D^1$ are precisely those that tend to $-\infty$ as $\alpha\to-\infty$.

\begin{theorem}\label{thm:spectrum1}
  The spectrum of the linear quantum graph $(\Lambda_n,B_n)$ consists of
a sequence of clusters $\Sigma_n^k\!=\!\Sigma_n^k(q,\alpha)$ of eigenvalues ($k\geq0$), interlaced with intermediate eigenvalues $\{\lambda_n^k\}_{k=1}^\infty$, plus two rogue eigenvalues $\lambda_n^=<\lambda_n^-$ lying below the lowest cluster $\Sigma_n^0$.
\begin{enumerate}
\item
Each cluster except the $0^\text{th}$ contains $n\!-\!1$ numbers that lie between two Dirichlet eigenvalues:
\begin{equation}
  \Sigma_n^k \subset (\lambda_D^k,\,\lambda_D^{k+1})
  \qquad \text{for }\, k\geq1.
\end{equation}
\item
The rogue cluster $\Sigma_n^0$ contains $n\!-\!2$ numbers, which are less than $\lambda_D^1$.  Each $\lambda\in\Sigma_n^0\cup\{\lambda_n^1\}$ satisfies
\begin{equation}\label{hi1}
  \lambda \;=\; -(b+1)^{-2} \alpha^2 + O(1)
  \qquad (\alpha\to-\infty).
\end{equation}
\rev{If $q(x)\equiv0$, then $\Sigma_n^0$ is contained in an interval of exponentially small length,}
\begin{equation}\label{hi2}
  |\Sigma_n^0| \;\lesssim\; \frac{8\alpha^2}{(b+1)^2} e^{-(b+1)^{-1}|\alpha|}
  \qquad  (\alpha\to-\infty).
\end{equation}
\item
The intermediate eigenvalues satisfy $\lambda_n^k\to\lambda_D^k\;(k\to\infty)$ and interlace the clusters:
\begin{equation}
  \Sigma_n^{k-1} < \lambda_n^k < \Sigma_n^k
  \qquad \text{for }\, k\geq1.
\end{equation}
\item
\rev{As $k\to\infty$, the set $\sqrt{\Sigma_n^k(\alpha,q)}$ tends to the set $\sqrt{\Sigma_n^k(0,0)}$,
 which is periodic in the sense that $\sqrt{\Sigma_n^{k+1}(0,0)}=\sqrt{\Sigma_n^k(0,0)}+\pi$ for $k\geq0$.}
For sufficiently large $k$,
\begin{equation}
  \sqrt{\Sigma_n^k} \;\subset\; \pi k + (\beta,\,\pi\!-\!\beta),
  \qquad \big[\,\beta := \cos^{-1}\big(2/(b^{-1/2}+b^{1/2})\big)\,\big].
\end{equation}
\item
The rogue eigenvalues satisfy, for $\alpha\to-\infty$,
\begin{align}
   \lambda_n^- &\;=\; -b^{-2}\alpha^2 + O(1)  \label{hi3}\\
   \lambda_n^= &\;=\; -\alpha^2 + O(1)\,.  \label{hi4}
\end{align}
\end{enumerate}
\end{theorem}

\begin{theorem}\label{thm:spectrum2}
The spectrum of the linear quantum graph $(\Lambda_n,\Bo_n)$ consists of a sequence of clusters $\So_n^k\!=\!\So_n^k(q,\alpha)$ of eigenvalues ($k\geq0$), interlaced with pairs of intermediate eigenvalues $\{\lo_{n-}^k,\lo_{n+}^k\}_{k=1}^\infty$, plus one rogue eigenvalue $\lo_n^=$ lying below the lowest cluster $\Sigma_n^0$.
\begin{enumerate}
\item
Each cluster except the $0^\text{th}$ contains $n\!-\!2$ numbers that lie between two Dirichlet eigenvalues:
\begin{equation}\label{ho1}
  \So_n^k \subset (\lambda_D^k,\,\lambda_D^{k+1})
  \qquad \text{for }\, k\geq1.
\end{equation}
\item
The rogue cluster $\So_n^0$ contains $n\!-\!2$ numbers, which are less than $\lambda_D^1$.  Each $\lambda\in\So_n^0\cup\{\lo_{n-}^1\}$ satisfies
\begin{equation}
  \lambda \;=\; -(b+1)^{-2} \alpha^2 + O(1)
  \qquad (\alpha\to-\infty).
\end{equation}
\rev{If $q(x)\equiv0$, then $\So_n^0$ is contained in an interval of exponentially small length}
\begin{equation}
  |\So_n^0| \;\lesssim\; \frac{8\alpha^2}{(b+1)^2} e^{-(b+1)^{-1}|\alpha|}
  \qquad  (\alpha\to-\infty).
\end{equation}
\item
The intermediate eigenvalues satisfy $\{\lo_{n-}^k,\lo_{n+}^k\}\to\lambda_D^k\;(k\to\infty)$ and interlace the clusters:
\begin{equation}
  \So_n^{k-1} < \lo_{n-}^k < \lambda_D^k < \lo_{n+}^k < \So_n^k
  \qquad \text{for }\, k\geq1.
\end{equation}
\item
\rev{As $k\to\infty$, the set $\sqrt{\So_n^k(\alpha,q)}$ tends to the set $\sqrt{\So_n^k(0,0)}$
 which is periodic in the sense that $\sqrt{\So_n^{k+1}(0,0)}=\sqrt{\So_n^k(0,0)}+\pi$ for $k\geq0$.}
For sufficiently large $k$,
%
\begin{equation}
  \sqrt{\So_n^k} \;\subset\; \pi k + (\beta,\,\pi\!-\!\beta),
  \qquad \big[ \,\beta := \cos^{-1}(2/(b^{-1/2}+b^{1/2})\, \big].
\end{equation}
\item
The rogue eigenvalue satisfies, for $\alpha\to-\infty$,
\begin{equation}
   \lo_n^= \;=\; -\alpha^2 + O(1)\,.
\end{equation}
\end{enumerate}
\end{theorem}

Having Theorems~\ref{thm:spectrum1} and~\ref{thm:spectrum2}, obtaining the spectrum of the quantum tree $(\Gamma_{\!n},A_n)$ is straightforward in view of the decomposition given in Theorem~\ref{thm:decomposition}.
The intervals containing the clusters for large $k$ (Part 4) become the large-$\lambda$ bands for infinite graphs (see Section~\ref{sec:infinite} and~\cite[\S5.2]{Solomyak2003}).  The intervals are valid for all $k$ when $q(x)\!\equiv\!0$ and $\alpha\!=\!0$.  This is evident from the graphical depiction, as $c(\lambda)\!=\!\cos\sqrt{\lambda}$ when $q(x)\!\equiv\!0$ and the components of $D_n(y,z)=0$ and $\Do_n(y,z)=0$ are straight vertical lines when $\alpha\!=\!0$.
The analysis via orthogonal polynomials below shows that the eigenvalues of $(\Gamma_{\!n},\Asym_n)$ and $(\Gamma_{\!n},\Aosym_n)$ interlace each other.
Quite general results of this kind are obtained by Schapotschnikow~\cite{Schapotsch2006}.

\subsection{Eigenfunctions}\label{sec:eigenfunctions}

The rescaled eigenfunctions $\{\tilde u_k = b^{k/2} u_k\}$ for the discrete reduction of $(\Lambda_n,B_n)$ or $(\Lambda_n,\Bo_n)$ satisfy the recurrence
\begin{equation}
  \tilde u_{k-1} + \tilde u_{k+1} = b^{-1/2} v(\lambda)\, \tilde u_k,
\end{equation}
with the important function $v(\lambda)$ being
\begin{equation}
  v(\lambda) \;=\; (b+1)c(\lambda) + \alpha s(\lambda).
\end{equation}
Thus the solutions are of the form
\begin{equation}\label{utilden}
  \tilde u_k \;=\; a_1 \xi^k + a_2 \xi^{-k},
  \quad\text{where}\;\;\;
  \xi + \xi^{-1} = b^{-1/2}v(\lambda).
\end{equation}
Depending on the eigenvalue $\lambda$, an eigenfunction is either in the oscillatory or exponential regime:
\begin{equation}\label{regimes}
\begin{aligned}
  \big| b^{-\frac{1}{2}} v(\lambda) \big| < 2\qquad &\text{oscillatory}\\
  \big| b^{-\frac{1}{2}} v(\lambda) \big| > 2\qquad &\text{exponential}
\end{aligned}
\end{equation}
The rescaled determinants $\tilde D_n = b^{-n/2}D_n$ satisfy the same recurrence as the~$\tilde u_n$.  In the oscillatory regime, the $\tilde D_n$ (and therefore also $D_n$) has clusters of roots; it is the shaded region in Fig.\,\ref{fig:DandS}.

The clusters of eigenvalues $\Sigma_n^k$ and $\So_n^k$ lie in the oscillatory region, except perhaps for the outer eigenvalues of a cluster.  This is made more precise with the connection to orthogonal polynomials below.  The rogue eigenvalues are in the exponential regime.  The eigenfunctions $u(x)$ for six eigenvalues are shown in Fig.\,\ref{fig:Efcn}.  The number of roots of the eigenfunctions can be obtained, even for general trees~\cite{Schapotsch2006} by the method of the Pr\"ufer angle.  The $n$-th eigenfunction has $n\!-\!1$ roots if they do not occur at the vertices.  This fact is also obtained as a result of a broader theory of the dependence of spectrum on vertex conditions~\cite[Theorem~6.4]{BerkolaikoKuchment2012a}.

The intervals of $\lambda$-values in the oscillatory regime ($ \big| b^{-\frac{1}{2}} v(\lambda) \big| < 2$) turn out to be the spectral bands of infinite trees~\cite[Lemma~3.2]{Carlson1997}; this is discussed in Section~\ref{sec:infinite} below.

\begin{figure}[ht]
\centerline{\includegraphics[width=0.86\linewidth]{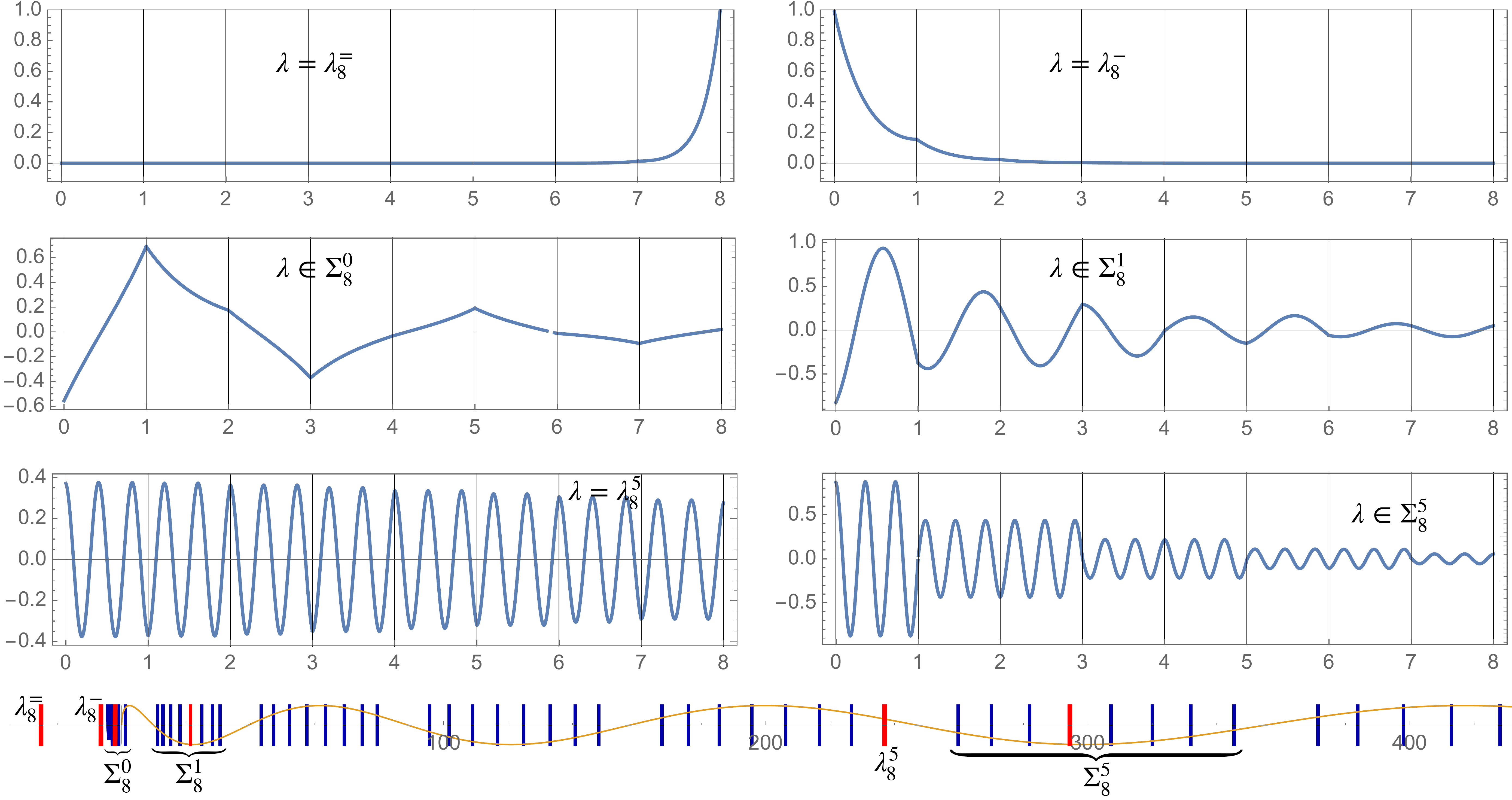}}
 \caption{\small For the tree $(\Lambda,B)$ of length $n$, the eigenfunction is plotted for different eigenvalues.  Here, $\alpha=-5$, $b=2$, and $q(x)\equiv0$.  The eigenvalues whose eigenfunctions are plotted are in red in the bottom diagram.  $\lambda_8^=$, $\lambda_8^-$, and $\lambda_8^5$ are in the exponential regime, and those in the clusters are in the oscillatory regime.  Observe that eigenfunctions in the oscillatory regime have a discrete decay rate of $b^{-1/2}$, which is compensated by the weight $b$ in the inner product~(\ref{inner}) for the linear-graph operator~$B$.}
 \label{fig:Efcn}
\end{figure}

\smallskip
{\bfseries\slshape Remarks on the connection to the Hill operator.}
The Hill operator is $-d^2/dx^2 + q(x)$ on the real line, where $q(x)$ is extended periodically.  There is a large body of literature on this operator.  Since $c^2-1=sc'$ when $q(x)$ is symmetric about $x=1/2$, the Hill discriminant is equal to $2c(\lambda)$.  Thus the parts of $\Spiral$ that lie inside $|y|\leq1$ correspond to spectral bands, and the complementary intervals are the gaps, or instability intervals of the Hill operator.  
The spectral bands for the Hill operator are wider than those of the semi-infinite tree operator for $\alpha\!=\!0$, discussed briefly in Sec.~\ref{sec:infinite}, since the oscillatory region (light-blue shaded region in the figures) does not extend up to $|y|=|c(\lambda)|=1$ for the tree.
The two points at which $\Spiral$ hits $y=1$ or $y=-1$ at the endpoints of a gap are the points where $s(\lambda)=0$ (a Dirichlet eigenvalue $\lambda_D^k$) or $c'(\lambda)=0$ (a Neumann eigenvalue $\lambda_N^k$).  These can come in either order, and they coincide exactly when the gap for the Hill operator closes into one point and $\Spiral$ intersects $y=\pm1$ tangentially.  The ordering of the intermediate eigenvalues $\lambda_n^k$, $\lo_{n\pm}^k$ and $\lambda_D^k$ depends on the slope of $\Spiral$ as it passes through $(\pm1,0)$ and the value of $\alpha$, which controls the slope of the extreme curves in the level set $D_n=0$, which pass through $(\pm1,0)$, and the slope of $\Do_n=0$ as it passes very near these points.

\section{Orthogonal polynomials}\label{sec:orthogonal}

A detailed analysis of the zero sets of $D_n(y,z)$ and $\Do_n(y,z)$ is made possible through their relationship to sets of orthogonal polynomials.  
Since the recurrence relation satisfied by both $D_n$ and $\Do_n$ involves $y$ and $z$ only through the composite variable
\begin{equation}
  v = (b+1)y + \alpha z,
\end{equation}
it is convenient to define two sequences of polynomials in~$v$,
\begin{align*} 
P_n(v) &=v\,P_{n-1}(v)-b\,P_{n-2}(v),\; \quad
 P_0 =1, \quad P_{-1} =0,\\
Q_n(v) &=v\,Q_{n-1}(v)-b\,Q_{n-2}(v), \quad
Q_0 =0, \quad Q_{-1}=1.
\end{align*}
In terms of these, one has
\begin{equation}\label{DPQ}
  \renewcommand{\arraystretch}{1.3}
\left.
\begin{array}{lcl}
   D_n(y,z) = D_n^\alpha(y,z) &=& \alpha z\, P_n(v) + (1-y^2)\, Q_n(v)\,,\\
  \Do_n(y,z) = \Do_n^\alpha(y,z) &=& P_n(v) + y\,Q_n(v)\,, \\
\end{array}
\right\}
  \;\;\text{with }\, v = (b+1)y + \alpha z\,.
\end{equation}
Because of Favard's Theorem~\cite[Theorem~4.4]{Chihara1978}, $\{Q_n\}$ and $\{P_n\}$ are sequences of orthogonal polynomials with respect to some measure $d\psi$ on the real line (see the appendix Sec.~\ref{sec:appendixmoments}).  Denote the roots of $P_n$ by $\left\{ v_{n1}, \dots, v_{nn} \right\}$ and the roots of $Q_n$ by
$\left\{ w_{n1}, \dots, w_{n,{n-1}} \right\}$, in increasing order.

\begin{proposition}\label{prop:PQD}
  $P_n$ and $Q_n$ are even or odd polynomials in $v$ and are related by $-bP_n(v)=Q_{n+1}(v)$.  For $n\geq1$, they have the form
\begin{equation}\label{PQleading}
  \renewcommand{\arraystretch}{1.1}
\left.
\begin{array}{lcl}
  P_n(v) &=& v^n - (n-1)b\, v^{n-2} + \dots \,,\\
  \vspace{-2.5ex} \\
  Q_n(v) &=& -b\,v^{n-1} + (n-2)b^2\, v^{n-3} + \dots\,,  
\end{array}
\right.
\end{equation}
in which the ellipses indicate lower-degree monomials.
They also admit the expressions
\begin{equation}\label{PQxi}
\begin{split}
  P_n(v) \;=\; b^{\frac{n}{2}}\frac{\xi^{n+1}-\xi^{-(n+1)}}{\xi-\xi^{-1}},
  \qquad
  &Q_n(v) \;=\; -b^{\frac{n+1}{2}}\frac{\xi^n-\xi^{-n}}{\xi-\xi^{-1}}, \\
  \text{where }\;\; \xi+\xi^{-1}&=b^{-\frac{1}{2}}v.
\end{split}
\end{equation}
The roots of $P_n$ ($n\geq0$) and the roots of $Q_n$ ($n\geq1$) are bounded by
\begin{equation}\label{rootsbound}
  \{v : P_n(v)=0 \text{ or } Q_n(v)=0 \}
  \;\subset\; \big( -2b^\frac{1}{2}, 2b^\frac{1}{2} \big)
  \;\subset\; \big( -(b+1), (b+1) \big),
\end{equation}
and thus they are in the oscillatory regime.
The roots of $P_n$ and $Q_n$ interlace each other, as do the roots of $P_n$ and $P_{n+1}$ and the roots of $Q_n$ and $Q_{n+1}$.

For $n\geq-1$,
\begin{equation}\label{symmetries}
\begin{aligned}
  D_n^\alpha(-y,-z) &= (-1)^{n+1} D_n^\alpha(y,z)\,, \\
  \Do_n^\alpha(-y,-z) &= (-1)^n\, \Do_n^\alpha(y,z)\,, \\
  D_n^{-\alpha}(-y,z) &= (-1)^{n+1} D_n^\alpha(y,z)\,, \\
  \Do_n^{-\alpha}(-y,z) &= (-1)^n \Do_n^\alpha(y,z)\,.
\end{aligned}
\end{equation}
\end{proposition}

\begin{proof}
The expressions for $P_n$ and $Q_n$ and the properties of $D_n$ and $\Do_n$ are simple to verify.
The interlacing of roots follows from $\{Q_n\}$ and $\{P_n\}$ being sequences of orthogonal polynomials.
To prove~(\ref{rootsbound}):
If $v>2b^{1/2}$, then $\xi\in\RR$, and thus $P_n(v)>0$ and $Q_n(v)<0$; $v=2b^{1/2}$ can be checked directly; then for $v\leq-2b^{1/2}$, use that these polynomials are even or odd.
\end{proof}

Figure\,\ref{fig:PQD}(left) illustrates how the straight-line zero sets of $P_n((b+1)y+\alpha z)$ and $Q_n((b+1)y+\alpha z)$ constrain the component curves of the zero set of $D_n(y,z)$ to lie in certain slanted linear strips.  The rightmost two of these curves are unconstrained for $y>1$, and in fact, on these curves, the value of $v$ becomes unbounded.  They are responsible for the two rogue negative eigenvalues that occur for large negative $\alpha$.  Theorem~\ref{thm:Dnzero} makes this observation precise.
Figure\,\ref{fig:PQD}(right) and Theorem~\ref{thm:Dnozero} apply analogously to the zero set of $\Do_n(y,z)$; in this case there is only one unconstrained curve, which is responsible for one rogue negative eigenvalue.

The three curves that are not constrained inside of linear strips are asymptotically close to straight lines as $y\to\infty$, as stated in Theorem~\ref{thm:roguecurves}.

\begin{figure}[ht]
\centerline{\includegraphics[width=0.9\linewidth]{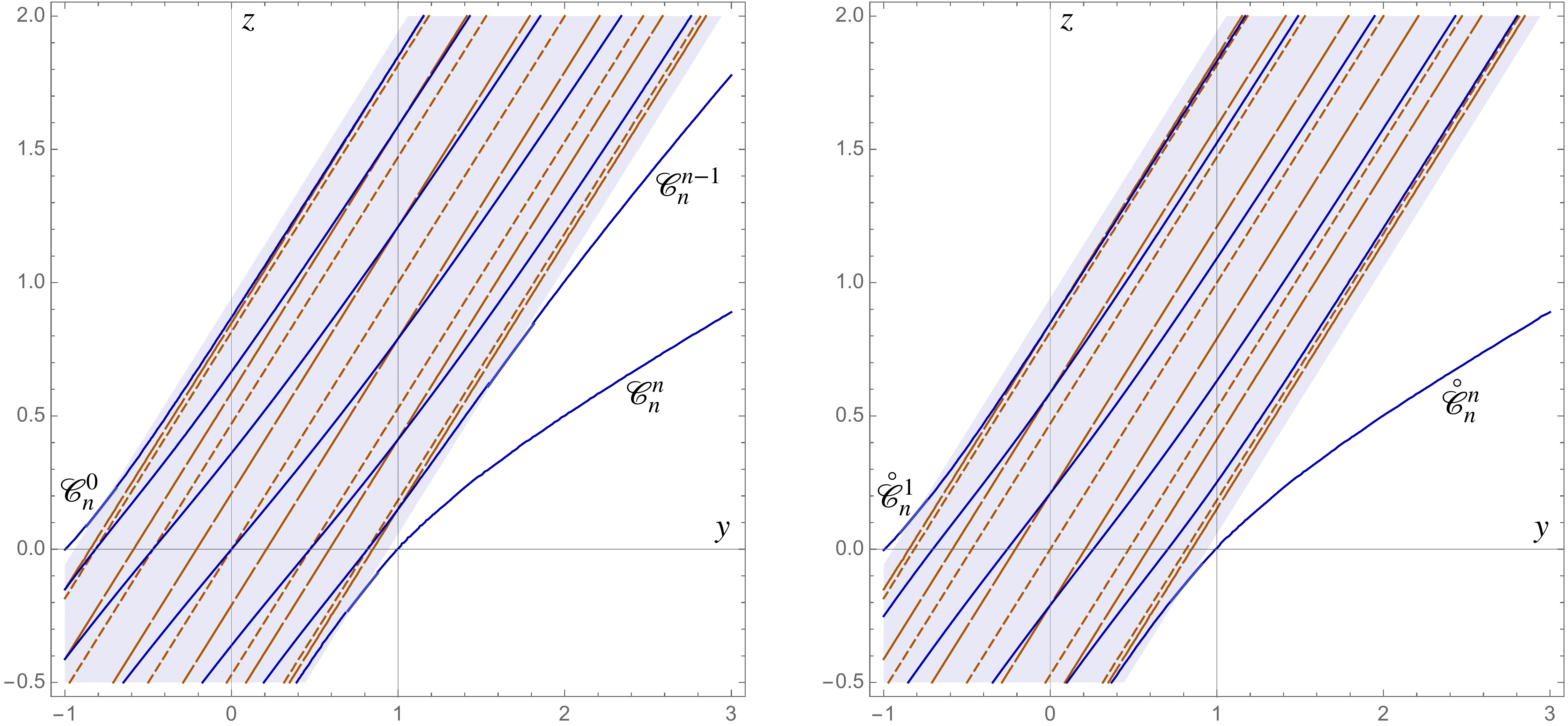}}
 \caption{\small Long-dashed straight lines are zero sets of $P_n(v)$, and short-dashed straight lines are zero sets of $Q_n(v)$, where $v=\alpha z+(b+1)y$.  Solid curves are zero sets of $D_n(y,z)$ (left) and $\Do_n(y,z)$ (right).  The zero sets of $P_n(v)$ and $Q_n(v)$ remain within the oscillatory regime (\ref{regimes}), which is shaded.
 Here, $b=2$, $\alpha=-3$, and $n=6$.}
 \label{fig:PQD}
\end{figure}

\begin{theorem}\label{thm:Dnzero}
The zero set $Z_n$ of $D_n$ consists of $n+1$ disjoint curves $\Cc_n^k$ for $0\leq k\leq n$,
\begin{equation}
  Z_n := \left\{ (y,z)\in\RR^2 : D_n(y,z)=0 \right\}
  = \bigcup\limits_{k=0}^{n} \Cc_n^k\,,
\end{equation}
and each $\Cc_n^k$ is the graph of an increasing function $y = g_n^k(z)$.
If $(y,z)\in\Cc_n^k$, then $(-y,-z)\in\Cc_n^k$, that is, the functions $g_n^k$ are odd.

\begin{enumerate}
\item
The curve $\Cc_n^0$ contains the point $(-1,0)$, and the curve $\Cc_n^n$ contains the point $(1,0)$.\label{Znintersection}
\item
For $0<k<n$, the curve $\Cc_n^k$ contains the point $(y,0)=\big(\frac{w_{nj}}{b+1},0\big)$, with
\begin{equation}\label{ybound1}
  \left| y \right| \,<\, \frac{2}{b^{1/2}+b^{-1/2}} \,<\, 1.
\end{equation}
\item\label{Cconstraints}
For $\alpha<0$ and $z>0$, the points $(y=g_n^k(z),z)\in\Cc_n^k$ lie inside the following slanted strips:
\begin{equation}
  \renewcommand{\arraystretch}{1.3}
\left.
  \begin{array}{rl}
    k=0: &
    \left\{
    \begin{array}{rcll}
      \hspace*{0pt} -(b+1)< & (b+1)y+\alpha z & < v_{n1}\hspace*{14.5pt} & \text{for }\, -1< y< 1 \\
      v_{n1}< & (b+1)y+\alpha z & < w_{n1} & \text{for }\, 1< y
    \end{array}
    \right. \\
    k=1,\dots,n-2: &
    \left\{
    \begin{array}{rcll}
      w_{nk}< & (b+1)y+\alpha z & < v_{n,k+1} & \text{for }\, -1< y< 1 \\
      \hspace*{10pt} v_{n,k+1}< & (b+1)y+\alpha z & < w_{n,k+1} & \text{for }\, 1< y
    \end{array}
    \right. \\
    k=n-1: &
    \left\{
    \begin{array}{rcll}
      \hspace*{7.7pt} w_{n,n-1}< & (b+1)y+\alpha z & < v_{nn}\hspace*{14.5pt} & \text{for }\, -1< y< 1 \\
      v_{nn}< & (b+1)y+\alpha z &  & \text{for }\, 1< y
    \end{array}
    \right. \\
    k=n: &
    \left.
    \begin{array}{rcll}
      \hspace*{23.5pt}b+1< & (b+1)y+\alpha z & \hspace*{38pt} & \text{for }\, 1< y\,.
    \end{array}
    \right.
  \end{array}
\right.
\end{equation}
\end{enumerate}
\end{theorem}

\begin{proof}
  Because of the recursion relation for the sequence $\{D_n(y,z)\}_{n=-1}^\infty$ and Favard's Theorem, for fixed~$y$ this is a sequence of orthogonal polynomials in $z$ (provided $\alpha\not=0$), and for fixed $z$, this is a sequence of orthogonal polynomials in $y$.  Since the roots of orthogonal polynomials are simple, it follows that for all $(y,z)$ such that $D_n(y,z)$ is zero, both $\partial D_n/\partial z (y,z)$ and $\partial D_n/\partial y(y,z)$ must be nonzero.  By the implicit function theorem, each component of the zero set of $D_n$ in the $yz$-plane is a strictly monotonic function $y=g(z)$.

Part 2 comes from (\ref{rootsbound}) and~(\ref{DPQ}).
The rest of the proof can be gotten from the relations~(\ref{DPQ}).
\end{proof}

\begin{theorem}\label{thm:Dnozero}
The zero set $\Zo_n$ of $\Do_n$ consists of $n$ disjoint curves $\Cco_n^k$ for $1\leq k\leq n$,
\begin{equation}
  \Zo_n := \left\{ (y,z)\in\RR^2 :\Do_n(y,z)=0 \right\}
  = \bigcup\limits_{k=1}^{n} \Cco^n_k\,,
\end{equation}
and each curve $\Cco_n^k$ is the graph of a monotonic function $y = \go_{nk}(z)$.
If $(y,z)\in\Cco_n^k$, then $(-y,-z)\in\Cco^n_k$, that is, the functions $\go_{nk}$ are odd.

\begin{enumerate}
\item
For $1\leq k\leq n$, the intersections of the curves $\Cco_n^k$ with the $y$-axis lie in the interval $(-1,1)$.
\item
For $1<k<n$, the point $(y,0)$ on the curve $\Cco_n^k$ satisfies
\begin{equation}\label{ybound2}
  \left| y \right| \,<\, \frac{2}{b^{-1/2}+b^{1/2}} \,<\, 1.
\end{equation}
\item\label{Coconstraints}
For $\alpha<0$ and $z>0$, the points $(y=g_n^k(z),z)\in\Cco_n^k$ lie inside the following slanted strips:
\begin{equation}
  \renewcommand{\arraystretch}{1.3}
\left.
  \begin{array}{rl}
    k=1: &
    \left\{
    \begin{array}{rcll}
      \hspace*{0pt} -(b+1)< & (b+1)y+\alpha z & < v_{n1}\hspace*{14.5pt} & \text{for }\,  y< 0 \\
      v_{n1}< & (b+1)y+\alpha z & < w_{n1} & \text{for }\, 0< y
    \end{array}
    \right. \\
    k=2,\dots,n-1: &
    \left\{
    \begin{array}{rcll}
      w_{n,k-1}< & (b+1)y+\alpha z & < v_{nk} & \text{for }\,  y< 0 \\
      \hspace*{10pt} v_{nk}< & (b+1)y+\alpha z & < w_{nk} & \text{for }\, 0< y
    \end{array}
    \right. \\
    k=n: &
    \left\{
    \begin{array}{rcll}
      \hspace*{7.7pt} w_{n,n-1}< & (b+1)y+\alpha z & < v_{nn}\hspace*{14.5pt} & \text{for }\,  y< 0 \\
      v_{nn}< & (b+1)y+\alpha z &  & \text{for }\, 0< y.
    \end{array}
    \right. \\
    \end{array}
\right.
\end{equation}
\end{enumerate}
\end{theorem}

\begin{proof}
For part 1: With $y=(b+1)v$ and $z=0$, (\ref{DPQ}) gives
\begin{equation}
  (b+1)\Do_n(y,0) \;=\; (b+1) P_n(v) + vQ_n(v).
\end{equation}
The numbers $R_n=(b+1)\Do_n(y,0)$ satisfy the same recursion as $P_n$ and $Q_n$, namely $R_{n+1}=vR_n-bR_{n-1}$, with initial values $R_{-1}=v$ and $R_0=b+1$, which yield $R_1=v$.
Fix $v\geq b+1$, so that $R_1\geq R_0$.  Recursively, one obtains $R_{n+1}\geq R_n$.  Thus, all roots of $\Do_n(y,0)$ are less than~$1$, and by symmetry, all roots are greater than~$-1$.
Part 3 can be gotten from the relations~(\ref{DPQ}).  Part 2 comes from Part 3 and (\ref{rootsbound}).
\end{proof}

\begin{remark}
  In Fig.\,\ref{fig:PQD}, the outer curves $\Cco_n^1$ and $\Cco_n^{n}$ may appear to intersect $(\pm1,0)$.  In fact they are very close, but, according to part 1 of Theorem~\ref{thm:Dnozero} do not actually hit those points.  A refining of the proof shows that $\Do_n(1,0)=1$.
\end{remark}

The next theorem gives the asymptotic behavior of the three curves that are not constrained in strips.

\begin{theorem}\label{thm:roguecurves}
If $b\geq2$ and $\alpha\not=0$, the two unconstrained components of the zero set of $D_n(y,z)$ have the following asymptotic behavior as $y\to\infty$ in the $yz$-plane.
\begin{equation}\label{roguecurves}
  \renewcommand{\arraystretch}{1.1}
\left.
\begin{array}{rcll}
    \Cc_n^n &:& \alpha z + y - y^{-1} + O(y^{-2}) = 0\; & (y\to\infty), \\
    \Cc_n^{n-1} &:& \alpha z + b\,y - b\,y^{-1} + O(y^{-2}) = 0\; & (y\to\infty).
\end{array}
\right.
\end{equation}
If $\alpha\not=0$, the one unconstrained component of the zero set of $\Do_n(y,z)$ has the following asymptotic behavior.
\begin{equation}\label{roguecurve}
    \Cco_n^n \;\; : \;\; \alpha z + y - y^{-1} + O(y^{-2}) = 0\; \qquad (y\to\infty).
\end{equation}
\end{theorem}

\begin{proof}
  From the form of the leading terms of $P_n$ and $Q_n$ given in (\ref{PQleading}), one obtains
\begin{equation}\label{ratio}
  \frac{-bP_n(v)}{vQ_n(v)} = 1 - b\,v^{-2} + O(v^{-4})
  \qquad (v\to\infty).
\end{equation}
From this and the relation $D_n(y,z) = \alpha z P_n(v) + (1-y^2) Q_n(v)$, the condition
$D_n(y,z)=0$ becomes
\begin{equation}\label{D0modified}
  \alpha z v - b\alpha\, z v^{-1} + z\,O(v^{-3}) + b\,y^2 - b \;=\; 0\,.
\end{equation}
Let us consider, for $v\to\infty$, solutions of the form
\begin{equation}\label{zvsy}
  \alpha z \;=\; c_1 y + c_0 + c_{-1}y^{-1} + O(y^{-2})
  \qquad
  (y\to\infty).
\end{equation}
Recall that $v=\alpha z + (b+1)y$.  The condition $c_1>-(b+1)$ guarantees that $v\to\infty$ as $y\to\infty$.
With the ansatz~(\ref{zvsy}), one obtains
\begin{equation}
  \renewcommand{\arraystretch}{1.1}
\left.
\begin{array}{lcl}
  \alpha z v &=& c_1(c_1+b+1)y^2 + c_0(2c_1+b+1)y
     + 2c_1c_{-1} + c_0^2 + c_{-1}(b+1) + O(y^{-1})\,, \\
  \vspace{-2.2ex} \\
  \alpha z v^{-1} &=& c_1(c_1+b+1)^{-1} + c_0(b+1)(c_1+b+1)^{-2}\, y^{-1} + O(y^{-2})\,.
\end{array}
\right.
\end{equation}
Inserting these into~(\ref{D0modified}) yields
\begin{equation*}
  \left[ c_1(c_1+b+1) + b \right] y^2 + c_0 \left( 2c_1 + b + 1 \right) y 
  - bc_1(c_1+b+1)^{-1}-b + 2c_1c_{-1} + c_0^2 + c_{-1}(b+1)
  \;=\; O(y^{-1})\,.
\end{equation*}
The vanishing of the $y^2$ term implies
$  c_1 = -b$
or
$  c_1 = -1$.
In either case, the vanishing of the $y$ term implies
$  c_0 = 0$
since
$  b\not=1$.
The vanishing of the constant term yields
\begin{equation}
  c_{-1} = \frac{b\left[ 1 + c_1(c_1+b+1)^{-1} \right]}{2c_1 + b + 1}\,,
\end{equation}
which equals $b$ when $c_1=-b$ and equals $1$ when $c_1=-1$.
Since both possible values of $c_1$ satisfy $c_1>-(b+1)$, one obtains asymptotic solutions of~(\ref{D0modified}) of the form~(\ref{zvsy}) as $v\to\infty$.  These necessarily coincide with the curves $\Cc_n^{n-1}$ and $\Cc_n^n$ because, according to Theorem~\ref{thm:Dnzero}, the value of $v$ remains bounded on the curves $\Cc_n^k$ for $k<n-1$.  This is the content of the statement~(\ref{roguecurves}) in the~theorem.

Again using~(\ref{ratio}) and the relation $\Do_n(y,z)=P_n(v)+yQ_n(v)$, the condition $\Do_n(y,z)=0$ is
\begin{equation}
  v\left( 1 - bv^{-2} + O(v^{-4}) \right) - by \;=\; 0
  \qquad
  (v\to\infty)\,.
\end{equation}
Using the ansatz~(\ref{zvsy}) yields
\begin{equation}
  (c_1+1)y + c_0 + \left( c_{-1} - b(c_1+b+1)^{-1} \right)y^{-1} + O(y^{-2}) \;=\; 0
  \qquad
  (y\to\infty)\,.
\end{equation}
This implies $c_1=-1$, $c_0=0$, and $c_{-1}=1$.  The resulting curve must coincide with $\Cco_n^n$ because, according to Theorem~\ref{thm:Dnozero}, the value of $v$ remains bounded on the curves $\Cco^n_k$ for $k<n$.
\end{proof}

\section{Analysis of eigenvalues and eigenfunctions---proofs}\label{sec:proofs}

This section provides a detailed analysis of the eigenvalues of the linear quantum graphs $(\Lambda_n,B_n)$ and $(\Lambda_n,\Bo_n)$ based on graphical analysis of the intersection of the curve $\Spiral$ with the level sets of $D_n(y,z)$ and~$\Do(y,z)$.  According to Theorem~\ref{thm:decomposition}, the spectrum of the quantum tree $(\Gamma_{\!n},A_n)$ is obtained from the spectra of these linear graphs.

\subsection{The curve $\Spiral$}

The curve $\Spiral$ in the $yz$-plane is parametrerized by analytic spectral functions of the potential $q(x)$, namely $(y,z)=(c(\lambda),s(\lambda))$ for $\lambda\in\RR$, shown in Fig.\,\ref{fig:DandSq0} and~\ref{fig:DandS}.
When the potential $q(x)=0$, the parameterization is 
\begin{equation}\label{q0}
  (y,z) = \bigg( \cos\sqrt{\lambda}\,,\;\,\frac{\sin\sqrt{\lambda}}{\sqrt{\lambda}} \,\bigg),
  \qquad \lambda\in\RR,
\end{equation}
which is shown in Fig.\,\ref{fig:DandSq0}.  For $\lambda>0$, it spirals inward toward the interval $[-1,1]$, hitting the endpoints at each pass, and for $\lambda<0$, it follows a curve asymptotic to $z\sim y/\log y$.  These behaviors are true for general potentials~$q(x)$, as detailed in the following proposition.

\begin{proposition}\label{prop:Spiral}
For any potential $q\in L^2[0,1]$:
\begin{enumerate}
  \item  For $\lambda\to\infty$, the functions $c(\lambda)$ and $s(\lambda)$ have asymptotics
\begin{equation}\label{asympinfty}
  \renewcommand{\arraystretch}{1.3}
\left.
\begin{array}{l}
  c(\lambda) \;=\; \cos\nu + O(\nu^{-1})\\
  s(\lambda) \;=\; \nu^{-1}\sin\nu + O(\nu^{-2})
\end{array}
\right.
\qquad (\lambda=\nu^2\to\infty).
\end{equation}
  \item
  If $q(1-x)=q(x)$, then
$\Spiral$ passes through $(-1,0)$ or $(1,0)$ every time it intersects the $y$-axis.  These intersections are transversal and occur exactly at the Dirichlet eigenvalues~$\lambda_D^k$.\label{Scrossing}
  \item
For $\lambda\to-\infty$,
the functions $c(\lambda)$ and $s(\lambda)$ have asymptotics
\begin{equation}\label{asympminfty}
  \renewcommand{\arraystretch}{1.3}
\left.
\begin{array}{l}
  c(\lambda) \;=\; e^\nu \big( \half + q_0\nu^{-1} + O(\nu^{-2}) \big) \\
  s(\lambda) \;=\; \nu^{-1} e^\nu \big( \half + q_0\nu^{-1} + O(\nu^{-2}) \big)
\end{array}
\right.
\qquad (\lambda=-\nu^2 \to -\infty),
\end{equation}
in which $q_0 = \quarter \int_0^1q(x)dx$.  
\end{enumerate}
\end{proposition}

\begin{proof}
For the asymptotics of $c(\lambda)$ and $s(\lambda)$, see \cite[Ch\,1]{PoschelTrubowitz1987} or \cite[\S1.1]{FreilingYurko2001}, for example.
For Part~\ref{Scrossing}:
The roots of $s(\lambda)$ are exactly the Dirichlet eigenvalues $\lambda_D^k$.
The Wronskian $W(c(x,\lambda),s(x,\lambda))$ is identically equal to $1$; and $s'=c$ whenever $q(x)=q(1-x)$.  This yields the relation $c^2-sc'=1$ and hence $c(\lambda)^2=1$ whenever $s(\lambda)=0$.  Furthermore, $ds/d\lambda(\lambda_D^k)\not=0$ since all the roots of $s(\lambda)$ are simple.
\end{proof}

\subsection{Eigenvalue clusters}\label{clusters}

The spiral curve $\Spiral$ passes alternately through $(-1,0)$ and $(1,0)$ at the Dirichlet eigenvalues~$\lambda_D^k$, according to Part~\ref{Scrossing} of Proposition~\ref{prop:Spiral}.  
And according to Part~\ref{Znintersection} of Theorem~\ref{thm:Dnzero}, the components $\Cc_n^0$ and $\Cc_n^n$ of the zero-set of $D_n$ pass through $(-1,0)$ and $(1,0)$, respectively.  Therefore, $\Spiral$ crosses each of the curves $\Cc_n^j$, for $1\leq j\leq n\!-\!1$, as $\lambda$ traverses each of the intervals $(\lambda_D^k,\lambda_D^{k+1})$, for $k\geq1$.  These intersections constitute the cluster $\Sigma_n^k$ of eigenvalues of $(\Lambda_n,B_n)$ for $k\geq1$.  \rev{This establishes Part\,1 of Theorem\,\ref{thm:spectrum1}.}

$\Spiral$ crosses each of the curves $\Cc_n^j$, for $1\leq j\leq n$ as $\lambda$ traverses the interval $(-\infty,\lambda_D^1)$.  This follows from the asymptotics of the curves $\Cc_n^n$ and $\Spiral$ as $y\to\infty$; the former obeys $z/y \sim -1/\alpha$ (Theorem~\ref{thm:roguecurves}), while the latter obeys $z/y \sim 1/\log y$ (Part~3 of Proposition~\ref{prop:Spiral}).  
Because of Part~\ref{Cconstraints} of Theorem~\ref{thm:Dnzero}, the components $\Cc_n^{n-1}$ and $\Cc_n^n$ break away from the cluster of components $\Cc_n^k$ ($0\leq k\leq n-2$).  The intersections of these two rightmost components with $\Spiral$ give the two rogue eigenvalues $\lambda_n^-$ and $\lambda_n^=$, while the intersections with the other components $\Cc_n^k$ ($1\leq k\leq n-2$) constitute the cluster~$\Sigma_n^0$.  \rev{All these eigenvalues are treated in Sec.\,\ref{subsec:rogue}.}

As for Theorem\,\ref{thm:spectrum2}, parts 1 and\,2 of Theorem~\ref{thm:Dnozero} guarantee that all intersections of $\Do(y,z)=0$ with the $y$-axis lie in the interval $(-1,1)$.
Therefore, $\Spiral$ crosses each of the component curves $\Cco_n^j$, for $1\leq j\leq n$ as $\lambda$ traverses each of the intervals $(\lambda_D^k,\lambda_D^{k+1})$, for $k\geq1$.  The intersections with the components $\Cco_n^j$, for $2\leq j\leq n\!-\!1$ constitute the cluster $\So_n^k$ of eigenvalues of $(\Lambda_n,\Bo_n)$ for $k\geq1$.  (The intersections with $\Cco_n^1$ or $\Cco_n^n$ are designated below as intermediate eigenvalues.)  

$\Spiral$ crosses each of the curves $\Cco_n^j$, for $1\leq j\leq n$ as $\lambda$ traverses the interval $(-\infty,\lambda_D^1)$, again because of the large-$y$ asymptotics. 
Because of Part~\ref{Coconstraints} of Theorem~\ref{thm:Dnozero}, the component $\Cco_n^n$ breaks away from the cluster of components $\Cco_n^k$ ($1\leq k\leq n-1$).  The intersections of this component with $\Spiral$ gives the rogue eigenvalue $\lo_n^=$, while the intersections with the other components $\Cco_n^k$ ($2\leq k\leq n-1$) constitute the cluster~$\So_n^0$.

We claim that $\Spiral$ intersects each component $\Cc_n^j$ or $\Cco_n^j$ exactly once at each pass, excluding the intersections at $(\pm1,0)$.  The argument is based on the analyticity of the operators $B_n=B_n(\alpha)$ and $\Bo_n=\Bo_n(\alpha)$ in~$\alpha$, proved in~\cite{BerkolaikoKuchment2012a}.  For any given $\lambda$-interval $(\lambda_D^k,\lambda_D^{k+1})$ or $(-\infty,\lambda_D^1)$, $-\alpha$ can be made great enough so that there is exactly one intersection of $\Spiral$ with each $\Cc_n^j$ or $\Cco_n^j$, which is transverse.  The corresponding eigenvalue is simple.  As $\alpha$ varies over $\RR$, no two intersections can collide since the curves $\{\Cc_n^j\}_{j=0}^n$ ($\{\Cco_n^j\}_{j=1}^n$) are disjoint.  Therefore the number of intersections between $\Spiral$ and $\Cc_n^j$ ($\Cco_n^j$) must remain one for all~$\alpha$.

The interlacing property of the roots of $P_n$ and $Q_n$ and how $D_n$ and $\Do_n$ are related to them (see (\ref{DPQ})) results in the interlacing of the eigenvalues in the cluster $\Sigma_n^k$ with those in the cluster $\So_n^k$.  This is evident from Fig.\,\ref{fig:PQD}.

\rev{Part 4 of Theorem~\ref{thm:spectrum1} is proved as follows.  The asymptotics of $\Spiral$ in Part~1 of Proposition~\ref{prop:Spiral}, show that, as $\lambda\to\infty$, $\Spiral$ is asymptotic to the curve $\lambda\mapsto(\cos\sqrt{\lambda\,},0)$, independently of~$q$.  By~(\ref{DPQ}), the continuous curves $\Cc_n^j$ are independent of $\alpha$ for $z\!=\!0$, and this proves the first statement of Part~4 (this is also seen cleanly in the graphical analysis).  The periodicity of $\sqrt{\Sigma_n^k(0,0)}$ occurs because the curves $\Cc_n^j$ are vertical lines for $\alpha\!=\!0$ and $c(\lambda)\!=\!\cos\sqrt{\lambda}$ for $q\!\equiv\!0$.
The curves $\Cc_n^j$ ($1\leq j\leq n\!-\!1$) intersect $\{z=0\}$ where $y$ satisfies (\ref{ybound1}).  This implies, for large enough $\lambda$ at the intersections, the bound $|\cos\sqrt{\lambda}|<2(b^{-1/2}+b^{1/2})^{-1}$, from which Part~4 follows.  Part~4 of Theorem~\ref{thm:spectrum2} is proved similarly.}

\subsection{Intermediate eigenvalues}\label{intermediate}

Recall that each segment of the curve $\Spiral$ between clusters $\Sigma_n^{k-1}$ and $\Sigma_n^k$, for odd $k\geq1$ hits $\Cc_n^0$ twice, unless it hits it tangentially (for even $k\geq2$, it hits $\Cc_n^n$). One of the intersections occurs at $\lambda_D^k$ and the other occurs at another intermediate eigenvalue~$\lambda_n^k$ of $(\Lambda_n,B_n)$, as long as the intersection is not tangential.  In this section, we argue that $\lambda_D^k$ is an eigenvalue of $(\Lambda_n,B_n)$ if and only if $\Spiral$ intersects $\Cc_n^0$ tangentially.  This happens when $\lambda_D^k$ and $\lambda_n^k$ coalesce; and in this case we define the intermediate eigenvalue $\lambda_n^k$ to be equal to $\lambda_D^k$.

To analyze the spectrum of $(\Lambda_n,B_n)$ around $\lambda_D^k$, we apply the ``dotted-graph" technique of~\cite[\S\,IV]{KuchmentZhao2019a}.  An extra vertex is placed in the interior of each edge, with the Neumann condition (Robin parameter equal to $0$) imposed; the potentials on the new smaller edges are inherited from the original edges, and one obtains a new quantum graph $(\dot\Lambda_n,B_n)$.  The Neumann condition guarantees continuity of functions and their derivatives across the new vertices, thus making them inconsequential for the operator, and thus we retain the notation $B_n$.  The graphs $(\Lambda_n,B_n)$ and $(\dot\Lambda_n,B_n)$ are essentially the same; particularly, their spectra are equal.  Furthermore, the extra vertex can be chosen so that $\lambda_D^k$ is not a Dirichlet eigenvalue of any of the edges of $\dot\Lambda_n$.  Thus the matrix $\dot M_n$ for the dotted graph, analogous to $M_n$ (\ref{Mn}), vanishes at $\lambda_D^k$ if and only if $\lambda_D^k$ is an eigenvalue of $(\Lambda_n,B_n)$.

Returning to the matrix $M_n$ (\ref{Mn}), notice that the Dirichlet-to-Neumann matrix (\ref{DtN}) is multiplied by $s(\lambda)$ when writing the Robin condition for each of the vertices.  This means that each row of $M_n=M_n(\lambda)$ has an extra factor of $s(\lambda)$ compared with the ``true" spectral matrix $\hat B_n(\lambda)$ for $(\Lambda_n,B_n)$, obtained by using the Dirichlet-to-Neumann map directly.  Thus
\begin{equation}
  \det(\hat B_n(\lambda)) \;=\; \frac{1}{s(\lambda)^{n+1}} D_n(\lambda).
\end{equation}
Next we use \cite[Proposition\,1]{FisherLiShipman2020a}, which tells us that the insertion of an extra vertex into an edge $e$ results in the determinant of the spectral matrix being multiplied by the factor
$s(\lambda)/(s_1(\lambda)s_2(\lambda))$, where $s_1(\lambda)$ and $s_2(\lambda)$ are the spectral $s$-functions for the sub-edges obtained by splitting $e$ by the new vertex.  By applying this fact to each of the $n$ edges of $\Lambda_n$, we obtain a relation between the determinants of $\hat B_n(\lambda)$ and $\hat {\dot B}_n(\lambda)$, the latter being the spectral matrix for the dotted graph:
\begin{equation}
  \det {\dot B}_n(\lambda) \;=\;  \frac{s(\lambda)^n}{s_1(\lambda)^ns_2(\lambda)^n} \det(\hat B_n(\lambda))
\end{equation}
Since $\lambda_D^k$ is an eigenvalue of $(\Lambda_n,B_n)$ if and only if $\det {\dot B}_n(\lambda_D^k)=0$, we find that $\lambda_D^k$ is an eigenvalue of $(\Lambda_n,B_n)$ if and only if $\lambda_D^k$ is a root of
\begin{equation}\label{newD}
  \frac{1}{s(\lambda)} D_n(\lambda).
\end{equation}
Recall that $\lambda_D^k$ is a root of $D_n(\lambda)$, so the pole of (\ref{newD}) at $\lambda_D^k$ is removable.

We have arrived at the conclusion that $\lambda_D^k$ is an eigenvalue of $(\Lambda_n,B_n)$ if and only if it is a multiple root of $D_n(\lambda)$.  This occurs when the two intersection points of $\Spiral$ and $\Cc_n^0$ coalesce into one tangential intersection. 

As discussed under ``eigenvalue clusters" above, the two intermediate eigenvalues of $(\Lambda_n,\Bo_n)$ between clusters $\So_n^k$ and $\So_n^{k+1}$ ($k\geq0$) come from the two intersections of $\Spiral$ with $\Cco_n^1$ or $\Cco_n^n$ near the Dirichlet eigenvalue~$\lambda_D^k$.

\rev{The limit $\lambda_n^k\to\lambda_D^k$ as $k\to\infty$ in Part~3 of Theorem~\ref{thm:spectrum1} is seen as follows.  The curve $\Spiral$ passes through $(-1,0)$ or $(1,0)$ at the energies $\lambda\!=\!\lambda_D^k$, and $\Cc_n^0$ passes through $(-1,0)$ and $\Cc_n^n$ passes through $(1,0)$.  The other points where $\Spiral$ intersects $\Cc_n^0$ or $\Cc_n^n$ correspond to the energies $\lambda_n^k$, and these coincide with $\lambda_D^k$ whenever the intersection is tangential.  
The asymptotics of $\Spiral$ for large positive $\lambda$, given by Part~1 of Proposition~\ref{prop:Spiral}, show that $\lambda_n^k$ tends to $\lambda_D^k$ as $k\to\infty$.  The corresponding statement in Part~3 of Theorem~\ref{thm:spectrum2} is proved similarly.}

\subsection{Large negative cluster and rogue eigenvalues}\label{subsec:rogue}

Quantifying the cluster $\Sigma_n^0$ and the rogue eigenvalues $\lambda_n^-$ and $\lambda_n^=$
as $\alpha\to-\infty$ requires finding the asymptotic relation between $\lambda$ and $\alpha$ at the intersection points of the curve $\Spiral$ with curves in the $yz$-plane that tend to straight lines as $y\to\infty$.  This is because the two unconstrained components $\Cc_n^{n-1}$ and $\Cc_n^n$ of $D(y,z)=0$ plus the slanted lines constraining the rest of the components all have the form $\alpha z + \beta y + \gamma + O(y^{-1}) = 0$, according to Theorems~\ref{thm:Dnzero} and~\ref{thm:roguecurves}.  Analogous reasoning applies to $\So_n^0$ and~$\lo_n^=$ by Theorems~\ref{thm:Dnozero} and~\ref{thm:roguecurves}.
Parts 2 and 5 of Theorems \ref{thm:spectrum1} and~\ref{thm:spectrum2} can be derived from these asymptotic forms together with the following Lemma.

\begin{lemma}\label{lemma:intersection}
The parameterized curve $\Spiral : (y,z)=(c(\lambda),s(\lambda))$ and the set\, $\{ (y,z) : \alpha z + \beta y + \gamma + O(y^{-1}) = 0 \}$ (with $\beta>0$) intersect in the $yz$-plane at values of $\lambda=-\nu^2<0$, where
\begin{equation}\label{nu1}
  \nu \;=\; -\beta^{-1}\alpha + O(\alpha^{-1})
  \qquad (\alpha\to-\infty).
\end{equation}
At this intersection, $c(\lambda)\sim \half e^{-\beta^{-1}\alpha}$ and $s(\lambda)\sim -\half\alpha^{-1}\beta e^{-\beta^{-1}\alpha}$ as $\alpha\to-\infty$.

When $q(x)\equiv0$, this value of $\nu$ can be refined to
\begin{equation}\label{nu2}
  \nu \;=\; -\beta^{-1}\alpha + 2\beta^{-2} \gamma\, \alpha\, e^{\beta^{-1}\alpha} + O(\alpha\, e^{2\beta^{-1}\alpha})
  \qquad
  (\alpha\to-\infty).
\end{equation}
\end{lemma}

\begin{proof}
Using the asymptotics~(\ref{asympminfty}), the two conditions $(y,z)=(c(\lambda),s(\lambda))$ and $\alpha z + \beta y + \gamma + O(y^{-1}) \;=\; 0$ together imply a relation between $\alpha$ and $\nu$,
\begin{equation}
  \alpha = c_1\nu + c_0 + O(\nu^{-1}),
\end{equation}
and the constants are computed to be
$c_1=-\beta$ and $c_0=0$, and hence $\alpha = -\beta\nu + O(\nu^{-1})$.  This implies
$\nu = -\beta^{-1}\alpha + O(\alpha^{-1})$.
Using this and (\ref{asympminfty}) produces the asymptotics of $c(\lambda)$ and $s(\lambda)$. 

In the case that $q(x)\equiv0$, one has
  $c(\lambda) = \half\, \left( e^\nu + e^{-\nu} \right)$
and  $s(\lambda) = \half\, \nu^{-1} \left( e^\nu - e^{-\nu} \right)$,
which yields
\begin{equation}
  \alpha\nu^{-1} + \beta + 2\gamma e^{-\nu} + O(e^{-2\nu}) = 0\,,
\end{equation}
from which one can deduce~(\ref{nu2}).
\end{proof}

\rev{To obtain part 5 of Theorem~\ref{thm:spectrum1}, recall that $\lambda_n^-$ is the $\lambda$-value at the intersection of $\Spiral$ with $\Cc_n^{n-1}$, and by (\ref{roguecurves}) of Theorem~\ref{thm:roguecurves}, $\Cc_n^{n-1}$ is $\alpha z+by+O(y^{-1})=0$.  Thus (\ref{nu1}) with $\beta=b$ results in $\lambda_n^-=-\nu^2=-b^{-2}\alpha^2 + O(1)$ as $\alpha\to-\infty$.
The curve $\Spiral$ intersects $\Cc^n_n$ at $\lambda^=_n$, and this time $\beta=1$, resulting in
$\lambda_n^==-\nu^2=-\alpha^2 + O(1)$.

To obtain (\ref{hi1}) of Part 2 of Theorem~\ref{thm:spectrum1}, observe that, according to Theorem~\ref{thm:Dnzero}, the curves $\Cc_n^k$ ($0\leq k\leq n-2$) lie in the region $|\alpha z+(b+1)y| \leq b+1$, and thus (\ref{nu1}) applies for all $\lambda\in\Sigma_n^0\cup\{\lambda_n^1\}$ with $\beta=b+1$.

The thickness result (\ref{hi2}) for $\Sigma_n^0$ when $q(x)\equiv0$ is obtained as follows.  For $z\geq0$, according to Theorem~\ref{thm:Dnzero}, the cluster $\Sigma_n^0$ lies in the region $-(b+1)<\alpha z+(b+1)y < b+1$, which is bounded by lines of the form $\alpha z+\beta y+\gamma=0$ with $\beta=b+1$ and $\gamma=-(b+1)$ on the right and $\beta=\gamma=b+1$ on the left.  According to Lemma~\ref{lemma:intersection}, $\Spiral$ intersects the right line at an energy $\lambda_-=-\nu_-^2$ and the left line at an energy $\lambda_+=-\nu_+^2$, where
\begin{equation}
  \nu_\pm \,=\, -(b+1)^{-1}\alpha \,\pm\, 2(b+1)^{-1}\alpha e^{(b+1)^{-1}\alpha} + O(\alpha e^{2(b+1)^{-1}\alpha}),
\end{equation}
so that
\begin{equation}
  \lambda_\pm \,=\, -(b+1)^{-2}\alpha^2 \,\pm\, 4(b+1)^{-2}\alpha^2 e^{2(b+1)^{-1}\alpha}
       + O(\alpha^2 e^{2(b+1)^{-1}\alpha}).
\end{equation}
$\Sigma_n^0$ lies in an interval of length less than $\lambda_+\!-\!\lambda_-$, which yields the result.

The corresponding results for $(\Lambda_n,\Bo_n)$ in Theorem~\ref{thm:spectrum2} are proved essentially identically.}

\subsection{Eigenfunctions}\label{efcns}

The coefficients $a_1$ and $a_2$ in~(\ref{utilden}) for the discretized eigenfunctions of $(\Lambda_n,B_n)$ can be found by using the first equation from~(\ref{Mn}).  We find that the eigenspace is spanned by
\begin{equation}\label{uk}
  u_k \;=\; 
   \big(b^{1/2}\xi - c\big) \big(b^{-1/2}\xi\big)^k - \big(b^{1/2}\xi^{-1} - c\big) \big(b^{-1/2}\xi^{-1}\big)^k,
\end{equation}
where, as before, $\xi+\xi^{-1}\!=\!b^{-\frac{1}{2}}v$ and $v=v(\lambda)=\alpha s(\lambda)+(b+1)c(\lambda)$.
Similarly, the eigenspace for a discretized eigenvalue of $(\Lambda,\Bo)$ is spanned by
\begin{equation}\label{uko}
  u_k \;=\; 
  \big( b^{-1/2}\xi \big)^k - \big( b^{-1/2}\xi^{-1} \big)^k.
\end{equation}
Whether $\lambda$ is an eigenvalue is determined by whether $\left\{ u_n \right\}$ satisfies the equation represented by the last row of $M_n$ or $\Mo_n$, which is of course equivalent to the vanishing of $D_n$ or~$\Do_n$.

Theorems~\ref{thm:Dnzero} and~\ref{thm:Dnozero} and the fact that all roots $v$ of $P_n$ and $Q_n$ lie in the oscillatory regime (see~(\ref{rootsbound})), allow us to make the following conclusion, wherein the visual aid of Fig.\,\ref{fig:PQD} is helpful.  Because of the symmetries (\ref{symmetries}) we can consider just $\alpha\leq0$ and $z\geq0$.  For $|y|\leq1$, the curves $\Cc_n^k$ for $1\leq k\leq n\!-\!1$ and $\Cco_n^k$ for $2\leq k\leq n\!-\!1$ are entirely within the oscillatory regime.  The curves $\Cc_n^0$ and $\Cco_n^1$ enter the oscillatory regime as $y$ increases.  At $y=1$, $\Cc_n^{n-1}$ is in the oscillatory regime, very close to the boundary of the two regimes, and as $y$ increases, it enters the exponential regime.  Corresponding conclusions are made concerning the regimes of the eigenfunctions for the various eigenvalues.  Six eigenfunctions are plotted in Fig.\,\ref{fig:Efcn}, and their regimes are stated in the caption.

The rogue eigenvalues $\lambda_n^-$, $\lambda_n^=$, and $\lo_n^=$ lie in the exponential regime when $\alpha$ is negative and large enough.  We are interested in the ratio of the coefficients of the modes in the eigenfunctions~$\{u_k\}$.  For $\lambda_n^-$ and $\lambda_n^=$, which lie on $\Cc_n^{n-1}$ and $\Cc_n^n$, Theorem~\ref{thm:roguecurves} gives us $\alpha s + \beta c = O(c^{-1})$ with $\beta\in\{b,1\}$.   Therefore, at these eigenvalues, $v=\alpha s + (b+1)c=O(c)$ as $c\to\infty$.  Taking $|\xi|\geq1$, we have $\xi\sim b^{-1/2}v$.  From this, we find that the ratio of coefficients is
\begin{equation}\label{coeffratio}
  \frac{b^{1/2}\xi - c}{b^{1/2}\xi^{-1} - c}
  \;=\; 
\renewcommand{\arraystretch}{1.1}
\left\{
\begin{array}{ll}
  O(c^{-2}) &\text{ for } \lambda_n^- \\
  1-b + O(c^{-2}) &\text{ for } \lambda_n^=.
\end{array}
\right.
\end{equation}
As for $\lo_n^=$, (\ref{uko}) shows that the ratio of the coefficients is $-1$.
Fig.\,\ref{fig:Efcn} illustrates these asymptotics, as the eigenfunction for $\lambda_n^=$ experiences exponential growth while that for $\lambda_n^-$ decays.

\section{Semi-infinite graphs}\label{sec:infinite}

We offer a few words about the semi-infinite quantum tree $(\Gamma_{\!\infty},A_\infty)$ with the same vertex conditions and potential as $(\Gamma_{\!n},A_n)$ but no terminal vertices.  Denote the corresponding linear graphs by $(\Lambda_\infty,B_\infty)$ and $(\Lambda_\infty,\Bo_\infty)$.  Eigenvalue clusters of the finite graphs, being in the oscillatory regime, survive the limit $n\to\infty$ and become spectral bands of the semi-infinite graphs, whereas in the exponential regime, eigenvalues of the semi-infinite graphs cannot be deduced from the $n\to\infty$ limit of the eigenvalues of the finite graphs.

As $n\to\infty$, the clusters $\Sigma_n^k$ and $\So_n^k$ of eigenvalues become spectral bands $\Sigma_\infty^k$ and $\So_\infty^k$ of the semi-infinite graphs.  These are of course the same bands found by Carlson~\cite[Theorem\,4.2]{Carlson1997} for trees with uniform degree $b+1$, that is, instead of having a distinguished root vertex with degree $b$ as our tree does.
\rev{The following theorem is from \cite{Carlson1997}, adapted to the notation in this paper.

\begin{theorem}[Thm.\,4.2 and Cor.\,4.5 of Carlson~\cite{Carlson1997}]\label{thm:carlson}
Denote by $T$ the quantum tree that is the same as $(\Gamma_{\!\infty},A_\infty)$ except with degree $b+1$ at the root vertex.  Define $\sigma_1\subset\RR$ to consist of all energies $\lambda$ in the closure of the oscillatory regime, that is, for which $|b^{-1/2}v(\lambda)|\leq 2$ (see Sec.~\ref{sec:eigenfunctions} above).
The spectrum of $T$ is
\begin{equation*}
  \sigma(T) \;=\; \sigma_1 \cup \sigma_{\!D}
\end{equation*}
with $\sigma_1$ being absolutely continuous and $\sigma_{\!D}$ being point spectrum.
\end{theorem}

The graphical apparatus developed in the present work provides a way to visualize the spectral transition from finite to semi-infinite trees.  The continuous spectrum $\sigma_1$ of the semi-infinite tree consists of a sequence of bands that are the $\lambda$ intervals corresponding to the intersection of $\Spiral$ with the light-blue shaded region in Figures \ref{fig:DandSq0}, \ref{fig:DandS}, and \ref{fig:PQD}.
  Theorem~\ref{thm:carlson} essentially says that intersections of $\Spiral$ with the curves $\Cc_n^k$ and $\Cco_n^k$ (corresponding to the clusters of eigenvalues of $(\Lambda_{\!n},A_n)$) fill up the part of $\Spiral$ that overlaps the shaded region.  The point spectrum of $T$, being $\sigma_{\!D}$, corresponds to the intersection of $\Spiral$ with the horizontal axis $\{z=0\}$.  Thus the bands interlace the point spectrum.
  We argue below that $\sigma_{\!D}$ is in the spectrum of $(\Lo_{\!\infty},B_\infty)$ and therefore also $(\Gamma_{\!\infty},A_\infty)$. 

The densities of states in each band can be computed from the asymptotic density of eigenvalues of the clusters as $n\to\infty$.  For a $\lambda$-interval $I$ in any given band, this measure $\mu$ is defined as the limiting ratio of the number of eigenvalues in $I$ to the number of eigenvalues in the band,
\begin{equation}
  \mu(I) \;:=\; \lim_{n\to\infty} \frac{\#(I,n)}{n}, 
\end{equation}
where $\#(I,n)$ is the number of eigenvalues of $(\Lambda_n,B_n)$ in~$I$.
To compute it, we need the density of roots of the determinants $D_n$ and $\Do_n$ as functions of $\lambda$, since these roots are precisely the eigenvalues.}  Because these roots interlace with those of $P_n(v(\lambda))$ and $Q_n(v(\lambda))$, we first need the asymptotic density of the roots $\{v_n^k\}_{k=1}^n$ or $\{w_n^k\}_{k=1}^{n-1}$ of the orthogonal polynomials $P_n(v)$ and $Q(v)$ as $n\to\infty$.  This density $d\mu(v)$ is computed in the Appendix~(\ref{density}) and has support in the $v$-interval $[-2b^{1/2},2b^{1/2}]$.  The clusters $\Sigma_n^k$ and $\So_n^k$, as $n\to\infty$, therefore have asymptotic density
\begin{equation}
  d\mu(v(\lambda)) \;=\; \frac{2b^{1/2}}{\pi \sqrt{4b-v(\lambda)^2}\,} \left|\frac{dv}{d\lambda}\right| d\lambda,
\end{equation}
with interval of support between values $\lambda$ where $v(\lambda)=\pm 2b^{1/2}$.
This can be written more explicitly when $v(\lambda)=\alpha s(\lambda)+(b+1)s(\lambda)$ is available.
But one can always compute the limits of high-energy clusters, independently of $q(x)$.

As $\lambda\to\infty$, $v(\lambda)\sim(b+1)\cos\sqrt{\lambda}$ and $dv/d\lambda \sim (b+1)(\sin\sqrt{\lambda})/(2\sqrt\lambda)$, and the density of states becomes
\begin{equation}
  d\mu(v(\lambda)) \;\sim\; \frac{\sin\sqrt{\lambda}}{\pi\sqrt\lambda}
       \frac{b^{1/2}(b+1)}{\mathrm{Re} \sqrt{4b-(b+1)^2\cos^2\sqrt{\lambda}}\,}d\lambda
       \qquad (\lambda\to\infty).
\end{equation}
The intervals where the square root is real are the spectral bands $\Sigma_\infty^k$ and $\So_\infty^k$,
\begin{align}
  &\sqrt{\Sigma_\infty^k} \;=\; \sqrt{\So_\infty^k} \;\sim\;  k \pi + [\beta,\,\pi-\beta] \qquad (k\to\infty)\\
  &\beta \,=\, \arccos\frac{2}{b^{\frac{1}{2}}+b^{-\frac{1}{2}}}.
\end{align}
If $q(x)\equiv0$ and $\alpha=0$, these are exact for all $k$.
These are precisely the bands in \S5.2 of Solomyak~\cite{Solomyak2003}.

For the exponential regime, $(\Lambda_\infty,B_\infty)$ or $(\Lambda_\infty,\Bo_\infty)$ has an eigenvalue at $\lambda\not\in\sigma_{\!D}$ exactly when the coefficient of the growing mode of $\{u_k\}$ vanishes in (\ref{uk}) or (\ref{uko}).  The Dirichlet eigenvalues must be examined separately.

For $(\Lambda_\infty,\Bo_\infty)$ the coefficients of $u_k$ never vanish, and thus there are no eigenvalues outside of $\sigma_{\!D}$.  Each $\lambda_D^j$ is in fact an eigenvalue, with eigenfunction equal to
\begin{equation}
  \left(\frac{(-1)^j}{b}\right)^{\!\!k} \!u_D^j(x)
  \quad
  \text{ for $x$ in the $k$th edge,}
\end{equation}
where $u_D^j(x)$ is the $j$th Dirichlet eigenfunction of $-d^2/dx^2+q(x)$ on $[0,1]$ ($1\leq j$).  These eigenvalues were found by Carlson~\cite[Theorem\,4.2]{Carlson1997} for infinite regular trees of uniform degree $b+1$; see also Solomyak~\cite[Theorem\,5.2]{Solomyak2003}.  In fact, the spectra of such trees coincides with that of $(\Lambda_\infty,\Bo_\infty)$;

For $(\Lambda_\infty,B_\infty)$, there are no Dirichlet eigenvalues in the spectrum by the following reasoning.  Because of the Robin condition at the root, $u_0\not=0$; and $c(\lambda_D^j)=\pm1$ so that $|u_D^j(0)|=|u_D^j(1)|\not=0$; and this, together with continuity, prohibits the decay required of an eigenfunction of the semi-infinite graph.
A value $\lambda\not\in\sigma_{\!D}$ is an eigenvalue whenever $b^{1/2}\xi(\lambda)-c(\lambda)=0$, with the convention that $|\xi|>1$.  Because of $\xi+\xi^{-1}=b^{-1/2}v$, this implies that $b(c-c^{-1})+\alpha s=0$.  But this last equation implies that either $b^{1/2}\xi(\lambda)=c(\lambda)$ or $b^{1/2}\xi(\lambda)^{-1}=c(\lambda)$.  Thus we seek roots of $b(c-c^{-1})+\alpha s$ and check whether $|b^{-1/2}c(\lambda)|>1$.  We do this numerically for potentials $q(x)=V\cchi_{[1/3,2/3]}(x)$.  As expected from the asymptotic~(\ref{coeffratio}), there is an eigenvalue near $\lambda_n^-$.  In addition, $b(c-c^{-1})+\alpha s$ has a root in each spectral gap and
the sign of $\log|b^{-1/2}c(\lambda)|$ alternates between gaps; thus there is one eigenvalue in every other spectral gap.

\section{Appendix on moments}\label{sec:appendixmoments}

By Favard's Theorem~\cite[Theorem~4.4]{Chihara1978}, $\{P_n\}$ is a sequence of orthogonal polynomials with respect to a measure $d\psi$, with $\psi$ being a (not strictly) increasing function on $\RR$:
\begin{equation}
  \int P_n(v)P_m(v)\,d\psi(v) \;=\; 0
  \qquad (m\not=n).
\end{equation}
For each $n>0$, let the positive numbers $A_{n1}, \dots, A_{nn}$ be the Gaussian quadrature weights,
and define the functions
\begin{equation*}
 \psi_n(v)= \left\{
  \begin{array}{lll}
0 &\;\;\text{ if } v < v_{n1}& \\
A_{n1}+A_{n2}+\dots+A_{np}&\;\;\text{ if }v_{np} \leq\;\; v <v_{n, p+1}&(1\leq p<n).\\
\mu_0 &\;\;\text{ if } v\geq v_{nn}& 
\end{array}
\right.
 \end{equation*}
The moments of $\psi$ and its approximants $\psi_k$ are
\begin{equation}
  \mu_k = \int v^k\,d\psi(v),
  \qquad
  \mu_k^{(n)} = \int v^k\,d\psi_n(v).
\end{equation}
Since the roots of $P_n$ lie between $-(b+1)$ and $(b+1)$,
\begin{equation}
  \mathrm{supp}(d\psi_n) \subset [-(b+1),b+1],
  \qquad
  \mathrm{supp}(d\psi) \subset [-(b+1),b+1].
\end{equation}
The $d\psi_n$ approximate $d \psi$ as measures in the sense that $d\psi_n$ produces integrals of polynomials of degree $k \leq 2n-1$ exactly.  Therefore,
\begin{equation}\label{equalmoments}
  \mu^{(n)}_k=\mu_k,
  \qquad 0\leq k \leq 2n-1.
\end{equation}
Since $P_n(v)$ and $Q_n(v)$ are either even or odd, their roots are symmetric about the origin, and thus all odd moments vanish,
\begin{equation}
  \mu_k=0
  \;\;\text{ and }\;\;
  \mu^{(n)}_k=0
  \qquad
  \text{($k$ odd).}
\end{equation}

The series for the ratio~(\ref{ratio}) used in the proof of Theorem~\ref{thm:roguecurves} can be refined as follows.

\begin{proposition}\label{prop:PQratio}
The ratio of $P_{n+1}(v)$ and $Q_{n+1}(v)$ satisfies
\begin{equation}
\dfrac{-bP_{n+1}(v)}{vQ_{n+1}(v)} 
   \;=\; 1-v^{-2}\sum\limits_{j=0}^{\infty}v^{-j}\mu_j^{(n)}.
\end{equation}
\end{proposition}
\noindent
The proof falls out of the identity
\begin{equation}
  \frac{-bP_{n+1}(v)}{vQ_{n+1}(v)}  \;=\; 1-\dfrac{1}{v} \int \dfrac{d \psi _n(t)}{v-t},
\end{equation}
which comes from $-bP_n(v)=Q_{n+1}(v)$ (Proposition~\ref{prop:PQD}) and~\cite[Chapter III, Theorem~4.3]{Chihara1978}.

\smallskip
This leads to a refinement of the asymptotics of the rogue curves in Theorem~\ref{thm:roguecurves}.
See~\cite{Wang2018a} for details.

\begin{theorem}\label{thm:cmu}
If $b\geq2$ and $\alpha\not=0$, the components $\Cc_n^n$ and $\Cc_n^{n-1}$ of $D_n(y,z)=0$ and the component $\Cco_n^n$ of $\Do_n(y,z)=0$ have the following asymptotic behavior as $y\to\infty$ in the $yz$-plane.
\begin{equation*}
\alpha z= -cy+cy^{-1}+c^{(n)}_{-2}y^{-2}+c^{(n)}_{-3}y^{-3}+c^{(n)}_{-4}y^{-4}+ \dots + c^{(n)}_{-k}y^{-k}+\dots
\end{equation*}  
The coefficient $c^{(n)}_{-k}$ depends on $\alpha$, $b$, and $\{\mu_j^{(n-1)}\}_{j=0}^{k-1}$ only.
\end{theorem}

Putting this together with~(\ref{equalmoments}) says that the coefficient $c^{(n)}_{-k}$ stabilizes when $n$ is large enough that $k < 2n$.  Furthermore, the expansions for the curves $\Cc^n_n$, {\itshape etc.} for two different values of $n$ are different.

The measure $d\mu$ can be computed easily from the expressions (\ref{PQxi}) of $P_n$ and $Q_n$.  
As a density, it is the limit of the density of roots of these polynomials as $n\to\infty$.  By putting $\xi=\exp(i\theta)$ we obtain $-b^{-(n+1)/2}Q_n=\sin n\theta/\sin\theta$.  Thus we seek the density of roots of $\sin n\theta$ as a function of $v$, which is
\begin{equation}
  R_n(v) \;=\; \sin \Big( n\arccos \frac{v}{2b^{1/2}} \Big),
  \qquad
  \left| v \right| < 2b^{1/2},
\end{equation}
which yields
\begin{equation}\label{density}
  d\mu(v) \;=\; \frac{2b^{1/2}\,dv}{\pi\sqrt{4b-v^2}},
  \qquad
  \left| v \right| < 2b^{1/2}.
\end{equation}

The following proposition refines the expression of $P_n$ and $Q_n$; its proof is omitted (see~\cite{Wang2018a}).

\begin{proposition}\label{prop:coefficients}
For $n\geq1$, $P_n(v)$ is the polynomial part of
\begin{equation*}
  v^n - (n-1)b\, v^{n-2} + \dfrac{(n-3)(n-2)}{2}b^2v^{n-4}-\dfrac{(n-5)(n-4)(n-3)}{6}b^3v^{n-6}+\dots\,,
\end{equation*}
and $Q_n(v)$ is the polynomial part of
\begin{equation*}
  -b\,v^{n-1} + (n-2)b^2\, v^{n-3} - \dfrac{(n-4)(n-3)}{2}b^3v^{n-5}+\dfrac{(n-6)(n-5)(n-4)}{6}b^4v^{n-7}+
 \dots\,,
\end{equation*}
in which the ellipses indicate lower-degree monomials.
\end{proposition}


\end{document}